\documentclass[journal]{IEEEtran}
\usepackage{amsmath,amssymb,amsfonts,mathtools}
\usepackage{amsthm}
\usepackage{algorithm}
\usepackage{algpseudocode}
\usepackage{graphicx}
\usepackage{booktabs}
\usepackage{multirow}
\usepackage{array}
\usepackage{xcolor}
\usepackage{colortbl}
\usepackage{pgfplots}
\pgfplotsset{compat=1.18}
\usepgfplotslibrary{groupplots}
\usetikzlibrary{arrows.meta,positioning,fit,calc,shapes.geometric,backgrounds}
\usepackage{url}
\usepackage[hidelinks]{hyperref}
\usepackage[capitalise,nameinlink]{cleveref}
\usepackage{cite}

\newtheorem{theorem}{Theorem}
\newtheorem{lemma}{Lemma}
\newtheorem{proposition}{Proposition}
\newtheorem{corollary}{Corollary}
\theoremstyle{definition}
\newtheorem{definition}{Definition}
\newtheorem{remark}{Remark}

\newcommand{\scheme}{\textsf{SILMARILS-MV}}
\newcommand{\silm}{\textsf{SILMARILS}}
\newcommand{\F}{\mathbb{F}_q}
\newcommand{\bfe}{\mathbf{e}}
\newcommand{\blam}{\boldsymbol{\lambda}}
\newcommand{\balpha}{\boldsymbol{\alpha}}
\newcommand{\bc}{\mathbf{c}}
\newcommand{\bb}{\mathbf{b}}
\newcommand{\bu}{\mathbf{u}}
\newcommand{\bv}{\mathbf{v}}
\newcommand{\bg}{\mathbf{g}}
\newcommand{\bh}{\mathbf{h}}
\newcommand{\bD}{\mathbf{D}}
\newcommand{\Fki}{\mathcal{F}_{\mathsf{KI}}}
\newcommand{\Pki}{\Pi_{\mathsf{DKI}}}
\newcommand{\Bad}{\mathsf{Bad}}
\newcommand{\Ideal}{\mathsf{Ideal}}
\newcommand{\Real}{\mathsf{Real}}
\newcommand{\Sign}{\mathsf{Sign}}
\newcommand{\Ver}{\mathsf{Verify}}
\newcommand{\AggVer}{\mathsf{AggVerify}}
\newcommand{\Agg}{\mathsf{Aggregate}}
\newcommand{\Res}{\mathsf{Resolve}}
\newcommand{\sh}[1]{[\![#1]\!]}
\newcommand{\rank}{\operatorname{rank}}
\newcommand{\spn}{\operatorname{span}}
\newcommand{\ind}{\mathbf{1}}
\newcommand{\T}{^{\mkern-1.5mu\mathsf{T}}}

\definecolor{cSIL}{HTML}{2a78d6}
\definecolor{cAGG}{HTML}{eb6834}
\definecolor{cED}{HTML}{1baf7a}
\definecolor{cD44}{HTML}{eda100}
\definecolor{cFAL}{HTML}{e87ba4}
\definecolor{cSLH}{HTML}{008300}
\definecolor{cD65}{HTML}{4a3aa7}
\definecolor{cRED}{HTML}{e34948}
\pgfplotsset{every axis/.append style={font=\footnotesize, grid=major, grid style={gray!20},
  tick label style={font=\scriptsize}, label style={font=\footnotesize},
  legend style={font=\scriptsize, draw=none, fill=none}, line width=0.9pt, mark size=1.8pt}}

\begin{document}

\title{\scheme: Information-Theoretically Secure Multi-Verifier Signatures with Dealer-Free Key Issuance and Linear Aggregation for Permissioned Blockchains}

\author{Khadijeh~Bagheri,
        Hassan~Khodaiemehr,
        and~Chen~Feng%
\thanks{K. Bagheri, H. Khodaiemehr and C. Feng are with the School of Engineering, Faculty of Applied Science, The University of British Columbia (Okanagan Campus), Kelowna, BC, Canada (e-mail: khadijeh.bagheri@ubc.ca; hassan.khodaiemehr@ubc.ca; chen.feng@ubc.ca).}%
\thanks{Preprint, October 2026. This is the full version of the paper, with all proofs and the complete evaluation.}}

\markboth{Preprint}%
{Bagheri \MakeLowercase{\textit{et al.}}: \scheme}

\maketitle

\begin{abstract}
Permissioned blockchains authenticate transactions, votes and quorum certificates with computationally secure signatures, and post-quantum replacements inflate consensus traffic by orders of magnitude. Information-theoretic alternatives have required a trusted authority, quantum channels, per-signature interaction, or single-use keys with bounded transfer. We present \scheme{}, a multi-verifier unconditionally secure signature in which a signature is the vector $\balpha=A\,\bfe(m)$ and each verifier holds a secret functional $\bc_i=\Gamma\T\blam(x_i)$ derived from \emph{virtual anchors} $\Gamma$ that nobody learns. Against unbounded adversaries corrupting up to $\omega$ verifiers and any signers, and observing every decision, an honest verifier accepts a forged or invalid signature with probability at most $Q/q+1/(q-1)$, a bound that is tight in its leading term. The scheme is strongly unforgeable, non-repudiable, transferable across unboundedly many hops, and linearly aggregatable. Keys are issued by the validators through perfectly secure MPC for $t<n/3$, without broadcast. To scale, we let the committee itself sign: in \emph{threshold notarization}, finality certificates have constant size $32(\omega+1)$ bytes, honest shares are released only on decided values, and per-validator issuance cost no longer depends on $n$. Arbiter-verified keys give clients $32$-byte signatures checked in batches, robust redistribution supports membership changes and share refresh, and verifiers can be enrolled retroactively. In C over $\mathbb{F}_{2^{255}-19}$, a certificate verifies in $2.4\,\mu$s. In an analytic normal-case model at $1000$ validators, threshold notarization needs $71$\,KiB of authentication traffic per validator per block, $22\times$ less than ML-DSA-44 within the admitted-verifier setting considered by \scheme{}. In a Rust HotStuff prototype with $100$ validators, \scheme{} commits $2.8\times$ more transactions per second than ML-DSA-44 with $67\times$ less cryptographic CPU, and Monte-Carlo attacks are consistent with the bounds.
\end{abstract}


\section{Introduction}\label{sec:intro}

\IEEEPARstart{D}{igital} signatures are the authentication backbone of permissioned blockchains. In Byzantine fault-tolerant (BFT) protocols such as PBFT~\cite{PBFT2002}, HotStuff~\cite{HotStuff2019,FastHotStuff2024} and Tendermint~\cite{Tendermint2018}, as well as in endorsement-based platforms such as Hyperledger Fabric~\cite{Fabric2018}, signatures serve three roles at once. They authenticate client transactions as these are relayed through gossip and mempools. They let a leader convince every replica that a quorum of $2t+1$ validators voted for a block, through a \emph{quorum certificate} (QC). And they provide transferable, non-repudiable evidence that is later used for accountability and slashing~\cite{Forensics2021,Polygraph2021,Casper2017}. The third role is the most demanding: evidence must stay convincing to third parties long after it was produced.

All signatures deployed today are computationally secure. Elliptic-curve schemes fall to Shor's algorithm~\cite{Shor1997,Aggarwal2018}. The NIST post-quantum standards, ML-DSA~\cite{FIPS204,Dilithium2018}, Falcon~\cite{Falcon2020} and SLH-DSA~\cite{FIPS205,SPHINCSplus2019}, replace the hardness assumption but keep the computational nature. They also inflate signatures from $64$\,B to between $0.7$\,KB and $17$\,KB, and BFT protocols carry these signatures in every vote and certificate. Post-quantum migration efforts for ledgers~\cite{FernandezCarames2020,PQFabric2021,Zhaolu2025} report precisely this overhead. A structural problem remains that no computational scheme removes: evidence signed today must remain unforgeable for as long as the ledger is audited; consequently, any future cryptanalytic progress against the underlying assumption retroactively weakens slashing and dispute records~\cite{Mosca2018}.

Information-theoretically (unconditionally) secure signatures (USS) remove computational assumptions entirely~\cite{ChaumRoijakkers1990,Hanaoka2000,Shikata2002,SwansonStinson2016}. Unlike public-key signatures, a USS must give each verifier a \emph{different secret} verification key, and it must explicitly handle disagreement between verifiers~\cite{SwansonStinson2016}. Every known approach pays a price that is hard to reconcile with a production ledger:
polynomial constructions~\cite{Hanaoka2000,Hanaoka2002PKC,Shikata2002} require a \emph{trusted authority} that generates all keys and can therefore forge; the universal-hashing scheme of Amiri~\emph{et al.}~\cite{Amiri2018} and quantum digital signatures (QDS)~\cite{Wallden2015,Yin2023NSR,Li2023OTUH,Kiktenko2022PRA,Pelet2022} remove the authority but consume key material per message, need secret or quantum channels for every signature, and support only a bounded number of transfers through decreasing verification levels; pseudo-signatures~\cite{PfitzmannWaidner1992,FitziWolf2004,Narayanan2023} and information-checking (IC) signatures~\cite{RabinBenOr1989,CDN2015} are tailored to broadcast inside MPC and support few transfers or require interaction.
Our own earlier framework, \silm~\cite{SILMARILS}, combines a compact algebraic designated-verifier (DV) mode with a three-party IC mode. The DV mode lets the verifier simulate signatures and is therefore not transferable. The three-party mode needs an interactive consistency phase with authenticated broadcast for every signature, and it supports exactly one transfer.

\subsection{Our Approach}
This paper asks whether a signature can be information-theoretically secure and at the same time non-interactive, multi-use, transferable without bound, aggregatable, and free of any single trusted authority, using only an assumption a BFT committee already makes: an honest supermajority $t<n/3$. We answer affirmatively with \scheme.

The design starts from one observation about our earlier construction. In \silm{} the verifier reconstructs a Shamir sharing with \emph{public} weights; hence anyone who learns the per-message receipt can solve the public linear system and forge~\cite{SILMARILS}. \scheme{} reverses this. The signature is a public vector $\balpha=A\,\bfe(m)$, where $\bfe(m)=(1,m,\dots,m^{\psi})\T$. The \emph{reconstruction weights} are secret and differ per verifier: verifier $V_i$ accepts iff $\bc_i\T\balpha=\bb_i\T\bfe(m)$ with $\bb_i=A\T\bc_i$. A coalition of $\omega$ verifiers learns $\omega$ linear functionals of $A$, which is never enough to predict the functional of an honest verifier. Conversely, a dishonest signer who knows $A$ does not know where the honest functionals point and therefore cannot craft a signature that some honest verifiers accept and others reject.

\textbf{One property, three uses.} Everything else in this paper follows from a single structural fact: signing and verification are linear in the verifier key, in the signature and in the signing key. Each linearity yields one component, and each component inherits the security theorem of the core scheme (\cref{thm:main}) instead of needing a new assumption.
(i)~\emph{Linear in the verifier key: dealer-free issuance.} We set $\bc_i=\Gamma\T\blam(x_i)$, where $\blam(x)=(1,x,\dots,x^\omega)\T$, $x_i$ is a public identifier and the \emph{virtual anchors} $\Gamma$ are held in shares by the committee and never reconstructed. Every verifier key is then a public linear function of shared secrets, issuance needs multiplications only to form $\Gamma A$, independently of the number of verifiers, and the point $x_0=0$ supplies a \emph{committee arbiter} whose key exists only in shared form.
(ii)~\emph{Linear in the signature: aggregation.} Signatures of different signers on one message add up, and a quorum certificate (QC) is a single $(\omega+1)$-vector.
(iii)~\emph{Linear in the signing key: threshold notarization.} The committee can hold $A$ in shares and sign \emph{as one signer}: each member releases its share of $A\bfe(m)$ once it has decided $m$, and any $t+1$ correct shares interpolate a certificate. The committee is then a single signer and its verifiers are external; hence neither the issuance cost nor the collusion bound $\omega$ grows with the committee.

The deployment modes of \cref{sec:ledger,sec:scale} are therefore not separate schemes. They differ only in who holds the signing matrix (a validator, a client, or the committee in shares) and who holds verifier keys (all validators, external auditors, or the arbiter alone), and every unforgeability claim reduces to \cref{thm:main}. The paper follows this line: the core scheme and its proof (\cref{sec:construction,sec:security}), issuance (\cref{sec:dki}), deployment and scaling (\cref{sec:ledger,sec:scale}), and evaluation up to $1000$ validators (\cref{sec:eval}).

\subsection{Contributions}
\Cref{tab:contrib} summarizes our contributions.
\begin{table}[!t]
\centering
\caption{Contributions of this paper.}
\label{tab:contrib}
\scriptsize
\setlength{\tabcolsep}{3pt}
\renewcommand{\arraystretch}{1.15}
\begin{tabular}{@{}>{\raggedright\arraybackslash}p{0.27\columnwidth}>{\raggedright\arraybackslash}p{0.69\columnwidth}@{}}
\toprule
\textbf{Contribution} & \textbf{Result}\\
\midrule
Construction (\cref{sec:construction}) & Signing, verification, linear aggregation and dispute resolution; a signature is $\omega+1$ field elements, signing costs $(\omega+1)(\psi+1)$ and verification $\omega+\psi+2$ multiplications; no interaction, no broadcast.\\
\midrule[0.25pt]
IT security (\cref{sec:security}) & $\Pr[\Bad]\le Q/q+1/(q-1)$ against unbounded adversaries that corrupt $\omega$ verifiers and any signers and observe every decision (\cref{thm:main}). Strong unforgeability, non-repudiation, unbounded-hop transferability, aggregate security, arbitration; the leading term of the bound, $\omega$ and $\psi$ are tight (\cref{prop:tight}); simulation-based security in the Fitzi framework (\cref{cor:fds}).\\
\midrule[0.25pt]
Dealer-free issuance (\cref{sec:dki}) & Perfectly secure MPC $\Pki$ for $t<n/3$ without broadcast; multiplications independent of the number of verifiers; semi-honest cost validated by simulation, malicious cost estimated.\\
\midrule[0.25pt]
Ledger integration (\cref{sec:ledger}) & Votes, aggregated QCs, transaction relays, information-theoretic equivocation evidence (\cref{thm:account}), epochs; three deployment modes (\cref{tab:modes}).\\
\midrule[0.25pt]
Scaling (\cref{sec:scale}) & Threshold notarization with soundness, uniqueness and robustness (\cref{thm:tn}) and per-validator issuance independent of $n$; $32$-byte arbiter-verified keys with batch arbitration (\cref{lem:batch}); retroactive enrollment (\cref{prop:enrol}); robust redistribution (\cref{lem:reshare}); comparison with threshold signatures (\cref{tab:threshold}).\\
\midrule[0.25pt]
Evaluation (\cref{sec:eval}) & C core versus Ed25519, ML-DSA, Falcon and SLH-DSA on one host; ledger model up to $n=1000$; Rust HotStuff prototype up to $n=100$; simulations with Byzantine members; Monte-Carlo checks of the bounds and an exhaustive check of batch arbitration.\\
\bottomrule
\end{tabular}
\end{table}

We are explicit about the costs (\cref{sec:discussion}). Verification keys are secret and verifiers must be admitted by the committee; key material is consumed per signature and refreshed every epoch; issuance needs private channels and an honest committee supermajority within each epoch. In the individual-vote mode, dealer-free issuance costs more total bandwidth than ML-DSA for committees larger than about ten validators. Threshold notarization removes this growth: at $n=1000$ it needs an estimated $71$\,KiB per validator per block in the normal case, against $1.5$\,MiB for ML-DSA-44, but its certificates attest to the committee's decision rather than to individual votes unless the accountable variant is added.

\section{Related Work}\label{sec:related}

\Cref{tab:relwork} summarizes prior work by area, and \cref{tab:related} compares the resulting signature families property by property. Two lines of work matter most. Information-theoretic signatures either rely on a trusted authority that can impersonate any signer~\cite{Hanaoka2000,Shikata2002,SwansonStinson2016}, or consume keys per message and degrade with every transfer~\cite{Amiri2018,Wallden2015}. Scalable certificates for BFT rely on computational threshold signatures~\cite{CKS2005,SBFT2019,ThRaccoon2024,Ringtail2025}, which \cref{tab:threshold} compares with threshold notarization. Kniep \emph{et al.}~\cite{KSW2026DQS} avoid post-quantum aggregation altogether: their distributed quorum signatures turn a certificate into a local event of a Bracha-style approval broadcast over ordinary post-quantum signatures, with constant-size messages and public verification, but a certificate convinces a party outside the protocol only through the $O(n)$ underlying signatures. Threshold notarization instead gives a constant-size, information-theoretic certificate that verifiers can forward, at the price of secret verifier keys. We follow the security model of Swanson and Stinson~\cite{SwansonStinson2016}, which distinguishes signatures \emph{produced} by the signer from signatures that merely \emph{pass} verification and makes dispute resolution and trust assumptions explicit.

\begin{table*}[!t]
\centering
\caption{Related work by area, and how \scheme{} relates to it.}
\label{tab:relwork}
\scriptsize
\setlength{\tabcolsep}{3pt}
\renewcommand{\arraystretch}{1.12}
\begin{tabular}{@{}>{\raggedright\arraybackslash}p{0.12\textwidth}>{\raggedright\arraybackslash}p{0.31\textwidth}>{\raggedright\arraybackslash}p{0.19\textwidth}>{\raggedright\arraybackslash}p{0.34\textwidth}@{}}
\toprule
\textbf{Area} & \textbf{Representative work} & \textbf{What it achieves} & \textbf{Gap, and relation to \scheme{}}\\
\midrule
IT authentication and arbitration & Wegman--Carter~\cite{WegmanCarter1981,Stinson1994}; A$^2$/multi-receiver codes~\cite{Simmons1990,DesmedtYung1990,DesmedtFrankelYung1992,SafaviNainiWang1999} & Unconditional authentication, also when sender and receivers distrust each other & No transferable, non-repudiable evidence to third parties; \scheme{} adds secret per-verifier functionals with a global validity predicate.\\
\midrule[0.25pt]
USS with trusted authority & Chaum--Roijakkers~\cite{ChaumRoijakkers1990}; Hanaoka \emph{et al.}~\cite{Hanaoka2000,Hanaoka2002PKC}; Shikata \emph{et al.}~\cite{Shikata2002,Shikata2009}; Swanson--Stinson~\cite{SwansonStinson2016}; group variants~\cite{Seito2009} & Polynomial USS with transferability; formal models with dispute resolution & The authority generates all keys and can forge; in \scheme{}, keys are issued by MPC with virtual anchors, and the signer does not learn the key polynomial.\\
\midrule[0.25pt]
USS without authority & Amiri \emph{et al.}~\cite{Amiri2016eprint,Amiri2018} & Universal-hash keys exchanged over secret channels & One message per key distribution; transfer bounded by degrading thresholds. \scheme{} is multi-use per index, with unbounded transfer.\\
\midrule[0.25pt]
Quantum digital signatures and quantum-secured ledgers & QDS~\cite{GottesmanChuang2001,Dunjko2014,Wallden2015,Amiri2016PRA,Arrazola2016}; one-time universal hashing~\cite{Yin2023NSR,Li2023OTUH}; multi-user~\cite{Pelet2022,Kiktenko2022PRA}; ledgers~\cite{Kiktenko2018QSB,Sun2019Entropy,Li2022QDPoS} & IT security through quantum key distribution & Quantum links, per-message keys, bounded transfer. \scheme{} needs no quantum hardware and can use QKD keys for its issuance channels.\\
\midrule[0.25pt]
Information checking and pseudo-signatures & IC~\cite{RabinBenOr1989,CDN2015}; pseudo-signatures~\cite{PfitzmannWaidner1992,PfitzmannWaidner1996,FitziWolf2004,Narayanan2023}; Fitzi framework~\cite{Fitzi2002,TyagiWatanabe2023} & Broadcast and MPC from correlated randomness & Interactive or bounded transfers; the three-party mode of \silm~\cite{SILMARILS} needs broadcast and allows one transfer. \scheme{} meets the same framework non-interactively (\cref{cor:fds}).\\
\midrule[0.25pt]
Designated-verifier signatures & DVS~\cite{JSI96}; strong and universal DVS~\cite{SBWP03,ZhangSusiloGuo2025}; two-party \silm~\cite{SILMARILS} & Deniability: the verifier can simulate signatures & Non-transferable by design; \scheme{} targets the opposite goal, transferable evidence.\\
\midrule[0.25pt]
Post-quantum signatures in ledgers & Surveys and systems~\cite{FernandezCarames2020,Aggarwal2018,PQFabric2021}; NIST standards~\cite{FIPS204,FIPS205,SP800208}; Falcon SNARGs~\cite{Zhaolu2025}; BLS and Fabric multi-signatures~\cite{Boneh2018BLS,Xiao2021TIFS}; forensics~\cite{Forensics2021,Polygraph2021} & PQ replacements and aggregation for consensus; accountability needs transferable evidence & Large signatures, costly aggregation, computational security; BLS is not PQ. \scheme{} gives small, linearly aggregatable IT evidence.\\
\midrule[0.25pt]
Threshold certificates and MAC-based BFT & BLS threshold~\cite{CKS2005,SBFT2019}; lattice threshold~\cite{ThRaccoon2024,Ringtail2025}; PBFT with MACs~\cite{PBFT2002} & Compact committee certificates; signature-free consensus & Computational and publicly verifiable, or non-transferable (MACs). TN combines MAC-authenticated agreement with IT certificates (\cref{sec:tn,tab:threshold}).\\
\midrule[0.25pt]
Post-quantum quorum certificates & Distributed quorum signatures (DQS)~\cite{KSW2026DQS} & Constant-size messages from ordinary post-quantum (e.g., hash-based one-time) signatures and a Bracha-style approval broadcast; certificates become local events & Computational. A certificate convinces a party outside the protocol only through the $O(n)$ underlying signed votes and approvals, and only some correct node can attribute the signers. TN gives a constant-size, transferable, IT certificate, at the price of secret verifier keys (\cref{tab:threshold}).\\
\midrule[0.25pt]
Honest-majority MPC and resharing & BGW/CCD~\cite{BGW1988,AsharovLindell2017,CCD1988}; linear MPC~\cite{DN2007,BTH2008,ATLAS2021,TurboPack2022}; packing~\cite{FranklinYung1992}; composition~\cite{KLR2010,Canetti2000}; OEC~\cite{BCG1993}; proactive sharing and redistribution~\cite{Herzberg1995,WongWangWing2002} & Perfect security for $t<n/3$; robust reconstruction; committee hand-over & General-purpose tools. $\Pki$ instantiates them for a depth-one circuit; \cref{lem:reshare} gives a perfectly private batch redistribution.\\
\bottomrule
\end{tabular}
\end{table*}

\begin{table*}[!t]
\centering
\caption{Qualitative comparison of signature families relevant to permissioned ledgers. ``IT'' = information-theoretic; ``TA'' = trusted authority; ``Hops'' = number of transfers between verifiers that preserve validity.}
\label{tab:related}
\scriptsize
\setlength{\tabcolsep}{2.5pt}
\resizebox{\textwidth}{!}{%
\begin{tabular}{@{}lllllllll@{}}
\toprule
\textbf{Family} & \textbf{Security basis} & \textbf{Setup / trust} & \textbf{Key reuse} & \textbf{Hops} & \textbf{Verifier forges?} & \textbf{Per-sig.\ interaction} & \textbf{Aggregation} & \textbf{Special hardware}\\
\midrule
ML-DSA, Falcon, SLH-DSA~\cite{FIPS204,Falcon2020,FIPS205} & computational (PQ) & PKI & unlimited & unbounded & no & none & SNARG-based~\cite{Zhaolu2025} & none\\
XMSS/LMS~\cite{SP800208} & computational (hash) & PKI, stateful & bounded & unbounded & no & none & no & none\\
Ed25519 / BLS~\cite{Ed25519,Boneh2018BLS} & computational (not PQ) & PKI & unlimited & unbounded & no & none & BLS: yes & none\\
DVS / SDVS~\cite{JSI96,ZhangSusiloGuo2025} & computational & PKI & unlimited & none (by design) & yes (simulates) & none & no & none\\
\silm{} 2-party~\cite{SILMARILS} & comp.\ (PRF, ROM/QROM) & session channel & unlimited & none (DV) & yes (simulates) & none & no & none\\
\silm{} 3-party~\cite{SILMARILS} & IT (IC) & auth.\ channels + broadcast & per signature & 1 & no & IC challenge phase & no & none\\
IC / pseudo-signatures~\cite{RabinBenOr1989,PfitzmannWaidner1996,FitziWolf2004} & IT & correlated randomness & bounded & bounded & no & IC: yes & no & none\\
USS with TA~\cite{Hanaoka2000,Hanaoka2002PKC,Shikata2002,SwansonStinson2016} & IT & trusted authority & bounded & supported & no & none & not studied & none\\
Amiri~\emph{et al.}~\cite{Amiri2018} & IT & secret channels, no TA & per message & bounded (levels) & no & key distribution & no & none\\
QDS~\cite{Wallden2015,Yin2023NSR,Li2023OTUH,Kiktenko2022PRA} & IT (quantum) & quantum links / QKD & per message & bounded (levels) & no & quantum distribution & no & quantum\\
BLS / lattice threshold~\cite{CKS2005,ThRaccoon2024,Ringtail2025} & comp.\ (BLS: not PQ) & DKG, PKI & unlimited & unbounded & no & lattice: signing rounds & threshold & none\\
DQS~\cite{KSW2026DQS} & computational (PQ) & PKI & one-time keys (WOTS-TS) & unbounded, $O(n)$ proof & no & approval broadcast & none (local events) & none\\
\textbf{\scheme{} (this work)} & \textbf{IT}$^{\ast}$ & \textbf{committee, $t<n/3$, point-to-point only} & \textbf{$\psi N$ per epoch} & \textbf{unbounded} & \textbf{no} & \textbf{none} & \textbf{linear and threshold} & \textbf{none}$^{\ddagger}$\\
\bottomrule
\multicolumn{9}{@{}l@{}}{$^{\ast}$For messages in $\F$; hash-then-sign for long messages adds only collision resistance. $^{\ddagger}$Issuance needs private channels, which are information-theoretic with QKD or pre-shared pads; $t+1$ colluding committee members of an epoch could forge for that epoch.}
\end{tabular}}
\end{table*}

\section{Preliminaries and System Model}\label{sec:model}

\subsection{Notation and Algebraic Tools}
Let $q$ be a prime and $\F$ the field with $q$ elements. Vectors are columns and $\T$ denotes transposition. For integers $\omega,\psi\ge 1$ and $x,m\in\F$ we write
\begin{align*}
\blam(x)&=(1,x,x^2,\dots,x^{\omega})\T\in\F^{\omega+1},\\
\bfe(m)&=(1,m,\dots,m^{\psi})\T\in\F^{\psi+1}.
\end{align*}
For a matrix $K$, $\operatorname{row}(K)$ is its row space. $\Pr[\cdot]$ is taken over all coins of the experiment, and $x\xleftarrow{\$}S$ denotes uniform sampling. We use two elementary facts throughout.

\begin{lemma}[Vandermonde independence]\label{lem:vdm}
If $x_0,\dots,x_{\omega}\in\F$ are pairwise distinct, then $\blam(x_0),\dots,\blam(x_\omega)$ form a basis of $\F^{\omega+1}$. Likewise, any $j\le\psi+1$ vectors $\bfe(m_1),\dots,\bfe(m_j)$ with distinct $m_\ell$ are linearly independent.
\end{lemma}

\begin{lemma}[Affine uniformity]\label{lem:affine}
Let $\mathcal{S}\subseteq\F^d$ be a non-empty affine subspace with linear part $\mathcal{L}$, let $X$ be uniform on $\mathcal{S}$, and let $f:\F^d\to\F$ be affine. If there is $\Delta\in\mathcal{L}$ with $f(X+\Delta)\ne f(X)$, then $f(X)$ is uniform on $\F$; in particular $\Pr[f(X)=y]=1/q$ for every fixed $y$.
\end{lemma}
\begin{proof}
Write $f(\mathbf{z})=\ell(\mathbf{z})+f_0$ with $\ell$ linear; then $\ell(\Delta)=f(X+\Delta)-f(X)\neq 0$. The sets $X_0+\{s\Delta: s\in\F\}$ partition $\mathcal{S}$ into lines of equal size. Along each line $f$ takes the values $f(X_0)+s\,\ell(\Delta)$, which is every element of $\F$ exactly once. Averaging over the lines gives the claim.
\end{proof}

\subsection{Participants, Channels and Adversary}
\textbf{Participants.} A \emph{committee} $\mathcal{P}=\{P_1,\dots,P_n\}$, the validators of a permissioned ledger, tolerates $t<n/3$ Byzantine members. A set of \emph{verifiers} $\mathcal{V}=\{V_1,\dots,V_\nu\}$, normally $\mathcal{V}=\mathcal{P}$, is indexed by public, pairwise distinct, non-zero identifiers $x_1,\dots,x_\nu\in\F$; hence $\nu\le q-1$. Committee share points (\cref{sec:tn,sec:reshare}) are likewise distinct and non-zero, which requires $n,n'\le q-1$, and we assume $\psi<q$. These conditions hold trivially for the field $q=2^{255}-19$ used throughout. A set of \emph{signers} $\mathcal{U}$ contains validators and clients. The \emph{committee arbiter} $V_0$ is a virtual verifier at point $x_0=0$: its key exists only in shared form inside the committee.

\textbf{Channels.} Key issuance (\cref{sec:dki}) assumes synchronous, private, authenticated point-to-point channels between committee members and from the committee to each signer and verifier. This is the standard model of information-theoretic MPC~\cite{BGW1988,CDN2015}; \cref{sec:discussion} discusses how to instantiate it. Signatures themselves travel over arbitrary public channels, and their verification is non-interactive.

\textbf{No broadcast channel.} We do \emph{not} assume a broadcast channel. Signing, aggregation and verification involve no broadcast at all: a signature is a single message that may be forwarded point-to-point, or through untrusted relays, any number of times. Key issuance and resolution use point-to-point channels only. Where the underlying MPC sub-protocols are stated with a broadcast channel (in their complaint and fault-detection phases), broadcast is emulated perfectly for $t<n/3$ by Byzantine agreement over point-to-point channels~\cite{PSL1980}, which the committee runs anyway as part of consensus. This contrasts with the three-party mode of \silm{} and with information checking~\cite{RabinBenOr1989,TyagiWatanabe2023}, where every signature requires an authenticated broadcast of the holder's challenge and of the consistency verdicts.

\textbf{Adversary.} The adversary $\mathcal{A}$ is computationally unbounded, malicious and static. It corrupts a set $T\subseteq\mathcal{P}$ with $|T|\le t$, a set of verifiers $C\subseteq\mathcal{V}$ with $|C|\le\omega$ and $T\cap\mathcal{V}\subseteq C$, and an \emph{arbitrary} subset $\mathcal{U}_C\subseteq\mathcal{U}$ of signers. When $\mathcal{V}=\mathcal{P}$ we choose $\omega\ge t$. Because security is information-theoretic and all honest messages are classical, the guarantees hold verbatim against quantum adversaries (\cref{rem:quantum}).

\subsection{Syntax}
\begin{definition}[Multi-verifier USS]
A multi-verifier unconditionally secure signature (MV-USS) scheme with parameters $(q,\omega,\psi,N)$ consists of:
$\mathsf{KeyIssue}$, an interactive protocol run by $\mathcal{P}$ (or, in the trusted-dealer variant, a dealer). It gives each signer $U\in\mathcal{U}$ a stateful signing key with $N$ \emph{indices} $k\in[N]:=\{0,\dots,N-1\}$, each usable for $\psi$ distinct messages, and gives each $V_i$ a secret verification key. $\Sign(U,k,m)\to\sigma$, which signs $m\in\F$ under index $k\in[N]$, or returns $\perp$ if index $k$ has already signed $\psi$ distinct messages. $\Ver_i(\sigma)\to\{0,1\}$, run by $V_i$. $\Agg$ and $\AggVer_i$, which combine and verify signatures of distinct signers on a common message. $\Res(\sigma)\to\{0,1\}$, the committee arbiter's verdict.
Arbitrary messages $M\in\{0,1\}^*$ are mapped to $m=H(M)\in\F$ by a collision-resistant hash (\cref{rem:hash}).
\end{definition}

\subsection{Security Model}
We follow the philosophy of Swanson and Stinson~\cite{SwansonStinson2016} and distinguish signatures that the signer \emph{produced} from signatures that merely \emph{pass} a verification algorithm. To quantify over adaptive behavior, the adversary observes every decision of every honest verifier.

\begin{definition}[Validity, authenticity, acceptability]\label{def:valid}
Each signer $U$ and index $k$ determine a secret \emph{signing matrix} $A_{U,k}$ (\cref{sec:construction}). An individual signature $\sigma=(U,k,m,\balpha)$ is \emph{valid} if $\balpha=A_{U,k}\bfe(m)$. It is \emph{authentic} if it was output by the signing oracle of an honest $U$. An aggregate $\sigma_S=(m,(U,k_U)_{U\in S},\balpha_S)$ is \emph{valid} if $\balpha_S=\sum_{U\in S}A_{U,k_U}\bfe(m)$, and \emph{authentic} if it is valid and every honest $U\in S$ has signed $(k_U,m)$. A signature or aggregate is \emph{acceptable} if it is authentic whenever it involves an honest signer, and valid otherwise.
\end{definition}

\begin{definition}[Experiment $\mathsf{Exp}^{\mathsf{mv\text{-}uss}}_{\mathcal{A}}$]\label{def:exp}
(i)~$\mathcal{A}$ fixes $T,C,\mathcal{U}_C$. Keys are issued, and $\mathcal{A}$ receives the keys of all corrupted signers and verifiers together with the view of $T$ during issuance. (ii)~$\mathcal{A}$ adaptively queries: \textsc{Sign}$(U,k,m)$ for honest $U$; \textsc{Ver}$(i,\sigma)$ and \textsc{AggVer}$(i,\sigma_S)$ for honest $V_i$, $i\in\{0,\dots,\nu\}\setminus C$; and \textsc{Res}$(\sigma)$. Each query returns the honest decision bit. Every decision honest parties take in a higher-level protocol is modeled as such a query. (iii)~$\Bad$ is the event that some query to an honest verifier (or to \textsc{Res}) returns $1$ on a signature or aggregate that is \emph{not acceptable}.
\end{definition}

\begin{definition}[$(Q,\varepsilon)$-security]\label{def:sec}
An MV-USS scheme is \emph{perfectly correct} if every valid signature and aggregate is accepted by every honest verifier and by $\Res$ with probability $1$. It is \emph{$(Q,\varepsilon)$-secure} if, for every adversary making at most $Q$ verification, aggregate-verification and resolution queries in total, $\Pr[\Bad]\le\varepsilon$.
\end{definition}

The single event $\Bad$ captures every requirement of~\cite{SwansonStinson2016} at once:
\emph{strong unforgeability}: an honest verifier never accepts a non-authentic signature of an honest signer, even one that modifies a genuine signature; \emph{non-repudiation}: a dishonest signer cannot get an \emph{invalid} signature accepted by an honest verifier; hence any accepted signature is valid and will also be upheld by the arbiter; \emph{transferability}: validity is a global predicate; if one honest verifier accepts, all honest verifiers accept (\cref{cor:transfer}); \emph{sound dispute resolution}: $\Res$ is an honest verifier.

\section{The \scheme{} Construction}\label{sec:construction}

\subsection{Design Rationale}\label{sec:rationale}
Three failure modes guided the design. We observed each of them in earlier algebraic attempts, including ours~\cite{SILMARILS}.
(i)~\emph{Public reconstruction is forgeable.} If acceptance is a \emph{public} linear predicate of the signature, as in the Shamir reconstruction of the \silm{} core once its receipt is known, then anyone can solve the predicate for a fresh message. Acceptance must therefore depend on verifier-specific \emph{secret} coefficients. (ii)~\emph{Shared verification secrets collapse to a MAC.} If a verifier's secret determines the signer's secret, the verifier can forge, and evidence is not transferable. A coalition of up to $\omega$ verifiers must learn too little to predict the signer's response to a new message. (iii)~\emph{Signer-dealt keys break agreement.} If the signer generates the verifiers' keys, it can give them inconsistent keys, and its signatures are accepted by some honest verifiers and rejected by others. Verification keys must be generated consistently by a party the signer does not control.
\scheme{} meets the three requirements with secret per-verifier linear functionals, uniformly random signing matrices, and committee-generated keys. No rank condition on the signing matrix is needed: the security proof uses only its uniform distribution.

\subsection{Ideal Key Issuance}
We first describe key issuance as an ideal functionality $\Fki$. It can be realized by a trusted dealer or, without any trusted party, by the committee protocol $\Pki$ of \cref{sec:dki}.

\begin{figure}[!t]
\fbox{\begin{minipage}{0.96\columnwidth}\small
\textbf{Functionality} $\Fki(q,\omega,\psi,N;\,x_0=0,x_1,\dots,x_\nu)$
\begin{enumerate}
\item \emph{Epoch anchors.} Sample $\Gamma\xleftarrow{\$}\F^{(\omega+1)\times(\omega+1)}$. For $i\in[\nu]$, send $\bc_i:=\Gamma\T\blam(x_i)$ to $V_i$. Keep $\bc_0:=\Gamma\T\blam(0)$ internally (the arbiter key).
\item \emph{Signer indices.} For every $U\in\mathcal{U}$ and $k\in[N]$ sample $A_{U,k}\xleftarrow{\$}\F^{(\omega+1)\times(\psi+1)}$ and send it to $U$. For $i\in[\nu]$, send $\bb_{i,U,k}:=A_{U,k}\T\bc_i\in\F^{\psi+1}$ to $V_i$, and keep $\bb_{0,U,k}$ internally.
\item \emph{Arbitration.} On input $\sigma=(U,k,m,\balpha)$ from the committee, return $1$ iff $\bc_0\T\balpha=\bb_{0,U,k}\T\bfe(m)$; aggregates are handled analogously.
\end{enumerate}
\end{minipage}}
\caption{Ideal key-issuance functionality. $\Gamma$ (the \emph{virtual anchors}) and the arbiter key are never output to anyone.}
\label{fig:fki}
\end{figure}

Write $M_{U,k}:=\Gamma A_{U,k}$. Then
\begin{equation}\label{eq:keypoly}
\bb_{i,U,k}\T=\bc_i\T A_{U,k}=\blam(x_i)\T\,\Gamma A_{U,k}=\blam(x_i)\T M_{U,k}.
\end{equation}
Verifier keys are therefore evaluations, at the public identifiers $x_i$, of the bivariate polynomial
\[
F_{U,k}(x,z)=\blam(x)\T M_{U,k}\,\bfe(z),
\]
of degree $\omega$ in $x$ and $\psi$ in $z$. The key of $V_i$ is $z\mapsto F_{U,k}(x_i,z)$. What distinguishes this from simply emulating the trusted authority of a polynomial USS~\cite{Hanaoka2000,Shikata2002} by MPC is that the \emph{signer} does not know $F_{U,k}$. If the signer's key determined the bivariate polynomial at public verifier identifiers, a dishonest signer could tailor signatures that pass at chosen identifiers and fail at others. It knows only $A_{U,k}$, which is related to $M_{U,k}=\Gamma A_{U,k}$ through the secret anchors (when $\Gamma$ is invertible, $A_{U,k}=\Gamma^{-1}M_{U,k}$ is a version of $M_{U,k}$ rotated by them). This is what prevents a dishonest signer from tailoring signatures to the public identifiers.

\subsection{Signing, Verification, Aggregation and Resolution}
\begin{algorithm}[!t]
\caption{\scheme{} algorithms}\label{alg:core}
\begin{algorithmic}[1]
\Statex \textbf{State of signer $U$:} $(A_{U,k})_{k\in[N]}$ and sets $\mathcal{M}_k\subseteq\F$ of messages signed under index $k$, initially empty.
\Procedure{Sign}{$U,k,m$}
  \If{$m\notin\mathcal{M}_k$ \textbf{and} $|\mathcal{M}_k|=\psi$} \Return $\perp$ \EndIf
  \State $\mathcal{M}_k\gets\mathcal{M}_k\cup\{m\}$; \quad $\balpha\gets A_{U,k}\,\bfe(m)$ \Comment{$(\omega+1)(\psi+1)$ mult.}
  \State \Return $\sigma=(U,k,m,\balpha)$
\EndProcedure
\Procedure{Verify$_i$}{$\sigma=(U,k,m,\balpha)$}
  \State \Return $[\,\bc_i\T\balpha=\bb_{i,U,k}\T\bfe(m)\,]$ \Comment{$\omega+\psi+2$ mult.}
\EndProcedure
\Procedure{Aggregate}{$m,(\sigma_U)_{U\in S}$} \Comment{all $\sigma_U$ on $m$}
  \State \Return $\sigma_S=(m,(U,k_U)_{U\in S},\textstyle\sum_{U\in S}\balpha_U)$
\EndProcedure
\Procedure{AggVerify$_i$}{$\sigma_S=(m,(U,k_U)_{U\in S},\balpha_S)$}
  \State \Return $[\,\bc_i\T\balpha_S=\sum_{U\in S}\bb_{i,U,k_U}\T\bfe(m)\,]$
\EndProcedure
\Procedure{Resolve}{$\sigma$} \Comment{committee MPC, \cref{sec:dki}}
  \State compute $\sh{\delta}=\sh{\bc_0}\T\balpha-\sh{\bb_{0,U,k}}\T\bfe(m)$; take a fresh preprocessed $\sh{\rho}$, $\rho\in\F^{*}$
  \State open $\delta\rho$; \Return $[\,\delta\rho=0\,]$
\EndProcedure
\end{algorithmic}
\end{algorithm}

\Cref{alg:core} lists the algorithms. Correctness is immediate: for a valid $\balpha=A_{U,k}\bfe(m)$, $\bc_i\T\balpha=\bc_i\T A_{U,k}\bfe(m)=\bb_{i,U,k}\T\bfe(m)$, and the aggregate check follows by linearity. In $\Res$, a sharing of a uniform non-zero mask $\rho$ is obtained in the standard way: generate random sharings $\sh{\rho},\sh{s}$, open $\rho s$, and retry if it is $0$~\cite{CDN2015}. Each attempt costs two $\mathsf{Rand}$ outputs, one multiplication and one opening, and succeeds with probability $(1-1/q)^2$. Masks are preprocessed in batches, and every mask is used for exactly one test: reusing $\rho$ on two non-zero residuals would reveal their ratio $\delta_1/\delta_2$. Then $\delta\rho$ is $0$ iff $\delta=0$, and otherwise it is uniform on $\F^*$; consequently, the arbiter key is never exposed. The online cost of a resolution is one multiplication and one opening; preprocessing the mask costs one more of each, plus two $\mathsf{Rand}$ outputs.

\textbf{Input validation.} The security analysis quantifies over well-formed objects, and implementations must enforce the same domain. $\Ver_i$, $\AggVer_i$ and $\Res$ reject an object whose field elements are not canonically encoded in $\{0,\dots,q-1\}$, whose index lies outside $[N]$ or was not issued to the stated signer for the object's \emph{declared signing epoch}, or whose message is not bound to that epoch and to the protocol domain. Keys are selected by the declared epoch. Whether an object from an earlier epoch is admissible is a separate policy of the operation: consensus messages must carry the current epoch, whereas historical verification and dispute resolution accept any epoch whose keys are retained. $\AggVer_i$ additionally rejects aggregates whose signer identities are not pairwise distinct; otherwise a repeated identity could inflate a quorum.

\begin{proposition}[Correctness]\label{prop:correct}
\scheme{} is perfectly correct: every valid, well-formed signature or aggregate that is admissible for the operation (see the input validation above) is accepted by every verifier holding keys output by $\Fki$ for its declared epoch, and by $\Res$.
\end{proposition}

\subsection{Parameters and Costs}
Typical choices are $q=2^{255}-19$, $\omega=t$ (the committee's fault bound), $\psi\in\{1,16,64\}$ messages per index, and $N$ indices per epoch. Per signer-index the signing key holds $(\omega+1)(\psi+1)$ field elements and each verifier stores $\psi+1$. A verifier additionally stores $\bc_i$ once per epoch. A signature carries its index (we budget $4$\,bytes) and $\omega+1$ field elements; an aggregate carries a signer bitmap, the list $(k_U)_{U\in S}$ of signer indices and $\omega+1$ field elements. Signing costs $(\omega+1)(\psi+1)$ multiplications plus $\psi-1$ to compute $\bfe(m)$, and verification costs $\omega+1$ plus $\psi+1$ multiplications. With $\psi=1$, a signature is \emph{one} matrix--vector product with a $2$-column matrix.

\textbf{Parameter selection.}\label{app:params}
\emph{Field.} With $q=2^{255}-19$ and, generously, $Q=2^{64}$ verification decisions over the system's lifetime, spread over at most $E=2^{64}$ epochs with independently sampled anchors, \cref{thm:main} and \cref{rem:epochs} give $\Pr[\Bad]\le2^{64}/(2^{255}-19)+2^{64}/(2^{255}-20)<2^{-189}$ for the core authentication budget; the MAC, decoder and redistribution terms of \cref{rem:budget} are separate. This bound holds against unbounded, including quantum, adversaries. \emph{Hashing.} For hashed long messages, the relevant assumption is collision resistance of the composite map $M\mapsto\mathrm{int}(H(M))\bmod q$, including distinct digests with the same residue. Because its range has $q\approx2^{255}$ elements, generic attacks find collisions with about $q^{1/2}\approx2^{127.5}$ classical or $q^{1/3}\approx2^{85}$ quantum queries~\cite{BHT1998}, whatever the digest length of $H$; SHA3-512 improves the uniformity of $m$ but not this margin. A larger margin requires a larger range: splitting a SHA3-512 digest injectively into three field elements (e.g., $254+254+4$ bits, each chunk below $q$) and signing them with the vector encoding of \cref{rem:hash} ($L=3$, one message per index, encoding dimension $4$) reduces the assumption to collision resistance of SHA3-512 itself. \emph{Collusion bound.} $\omega$ equals the number of verifiers that may be corrupted, i.e., $t$ for a BFT committee of $n=3t+1$. \emph{Capacity.} $\psi\in[1,64]$ balances signing time, which is linear in $\psi$, against verifier storage of $(\psi+1)/\psi$ elements per signature. \emph{Epochs.} $N$ indices per signer cover the epoch's expected signatures; for consensus, $N\approx(\text{blocks per epoch})/\psi$.
For example, at $n=31$, $\omega=10$, $\psi=16$ and a one-hour epoch of $3600$ blocks, each validator stores $3.7$\,MiB of verifier keys and $1.3$\,MiB of signing keys.

\begin{remark}[Message space]\label{rem:hash}
Messages are elements of $\F$. Longer messages are hashed, $m=H(M)$, with a collision-resistant $H$ such as SHA3-256 reduced modulo $q$. A forger then either breaks collision resistance or triggers $\Bad$; hence security degrades additively by $\mathsf{Adv}^{\mathrm{CR}}_H$, and \emph{only} collision resistance is used, never a random oracle. This term is a computational caveat: an unbounded adversary that chooses the hashed messages can find collisions, and the guarantee is then information-theoretic only for the injective encoding below. In BFT protocols, votes already sign a block digest, and this adds no new assumption. A purely information-theoretic alternative for messages in a fixed finite domain $M\in\F^L$ replaces $\bfe(m)$ by $\bfe(M)=(1,M_1,\dots,M_L)\T$ and allows one distinct message per index: any second signature under the same index is refused. The security proof of \cref{thm:main} applies with $\psi=1$ messages per index, because $(1,M)\T$ and $(1,M')\T$ are independent for $M\neq M'$. The encoding dimension, however, is $L+1$ rather than $\psi+1=2$: then $A_{U,k}\in\F^{(\omega+1)\times(L+1)}$ and $\bb_{i,U,k}\in\F^{L+1}$, and signing, verification, storage and issuance costs scale with $L+1$. The encoding must bind the protocol domain, epoch, slot and value within the declared message space; arbitrarily long messages cannot be encoded injectively into a fixed dimension.
\end{remark}

\begin{remark}[State]\label{rem:state}
As in XMSS/LMS~\cite{SP800208}, the signer keeps a counter. In a BFT ledger, each validator keeps one counter per message type and uses index $k=\lfloor j/\psi\rfloor$ for its $j$-th signature; signing two different messages for the same consensus slot is equivocation (\cref{sec:ledger}). \Cref{prop:tight} shows that exceeding $\psi$ messages on one index reveals that index's key. The damage is confined to that index and to the misbehaving signer.
\end{remark}

\section{Security Analysis}\label{sec:security}

We analyze \scheme{} in the $\Fki$-hybrid model. \Cref{sec:dki} replaces $\Fki$ by the committee protocol $\Pki$ with no loss. Throughout, $C$ is the set of corrupted verifiers ($|C|\le\omega$), and we write $\Lambda_C=[\blam(x_s)]_{s\in C}\in\F^{(\omega+1)\times|C|}$ and $K=\Lambda_C\T\Gamma$. The rows of $K$ are exactly the corrupted functionals $\bc_s\T$, $s\in C$. Since the arbiter point $x_0=0$ differs from every $x_s$, \cref{lem:vdm} gives
\begin{equation}\label{eq:indep}
\blam(x_i)\notin\spn\{\blam(x_s):s\in C\},\quad\forall\, i\in\{0,\dots,\nu\}\setminus C.
\end{equation}

\subsection{Main Theorem}
\begin{theorem}[Security of \scheme]\label{thm:main}
Let $\mathcal{A}$ be any computationally unbounded adversary that corrupts at most $\omega$ verifiers, an arbitrary set of signers, and makes at most $Q$ queries to honest verifiers (individual, aggregate and resolution queries together). In the experiment of \cref{def:exp}, in the $\Fki$-hybrid model,
\[
\Pr[\Bad]\;\le\;\frac{Q}{q}+\frac{1}{q-1}.
\]
Hence \scheme{} is $(Q,\,Q/q+1/(q-1))$-secure and perfectly correct.
\end{theorem}

The additive term $1/(q-1)$ bounds the probability that the uniformly sampled anchor matrix $\Gamma$ is singular. It is paid once per anchor matrix, not per query (\cref{rem:epochs}); \cref{rem:invertible} explains how it could be removed.

The proof has three steps. \emph{(i)}~We replace every verification answer by an \emph{ideal} answer that the adversary could compute itself; this does not change $\Pr[\Bad]$. \emph{(ii)}~In the resulting game, the posterior distribution of all secrets is uniform on a product of affine subspaces. \emph{(iii)}~Every individual query triggers $\Bad$ with probability at most $1/q$, by \cref{lem:affine} applied along an explicit direction in either the anchor matrix $\Gamma$ or a signing matrix $A_{U,k}$.

\begin{proof}
\emph{Step 1 (ideal answers).} Let $G_0$ be the real experiment. In $G_1$, every query $j$ to an honest verifier is answered with
$\Ideal_j:=1$ iff the queried object is acceptable (\cref{def:valid}) at the time of the query,
instead of the real decision $\Real_j$. Both games use the same coins. Acceptable objects are valid, and by \cref{prop:correct} valid objects are accepted; hence $\Ideal_j=1\Rightarrow\Real_j=1$. Therefore $\Real_j\neq\Ideal_j$ iff $\mathsf{bad}_j:=(\Real_j=1\wedge\Ideal_j=0)$. The games proceed identically until the first $j$ with $\mathsf{bad}_j$, and that event is precisely $\Bad$ in $G_0$. By the fundamental lemma of game playing, $\Pr_{G_0}[\Bad]=\Pr_{G_1}[\bigcup_j\mathsf{bad}_j]$, where in $G_1$ each $\mathsf{bad}_j$ is evaluated on the actual keys but never influences the adversary.

\emph{Step 2 (posterior).} Let $\tau_j$ be the adversary's view in $G_1$ just before query $j$. It consists of its random tape; $K$; the corrupted key polynomials $B_{U,k}:=K A_{U,k}$ for every $(U,k)$; the matrices $A_{U,k}$ of corrupted signers; and the answers to earlier queries. In $G_1$, a signing answer equals $A_{U,k}\bfe(m)$. A verification answer $\Ideal$ is a deterministic function of earlier answers (for honest signers, acceptability depends only on the signing log) and of the corrupted signers' matrices, both of which are part of $\tau_j$. The matrices of corrupted signers are sampled independently of those of honest signers. Hence the likelihood of $\tau_j$, given the honest secrets $\mathbf{s}=(\Gamma,(A_{U,k})_{U\notin\mathcal{U}_C})$, equals a factor independent of $\mathbf{s}$ times the indicator $\ind[\mathbf{s}\in\mathcal{S}(\tau_j)]$, where
\begin{align*}
\mathcal{S}(\tau_j)=&\;\{\Gamma:\Lambda_C\T\Gamma=K\}\;\times\\
&\prod_{U\notin\mathcal{U}_C,\,k}\{A:\;KA=B_{U,k},\;A\,E_{U,k}=R_{U,k}\}.
\end{align*}
Here $E_{U,k}=[\bfe(m)]_{m\in\mathcal{M}_k}$ and $R_{U,k}$ is the matrix of returned signatures. The fixed matrix $K$ read from $\tau_j$ appears in the second factor; therefore $\mathcal{S}(\tau_j)$ is a \emph{product of affine subspaces}. With a uniform prior, the posterior of $\mathbf{s}$ given $\tau_j$ is uniform on $\mathcal{S}(\tau_j)$, and the factors are independent.

\emph{Step 3 (per-query bound).} Fix $\tau_j$ and let query $j$ go to an honest verifier $V_i$ (with $i=0$ for $\Res$; there the mask $\rho\neq 0$ makes the decision exact). We use two kinds of shifts.

\emph{(A) $\Gamma$-shift.} Suppose the real decision is $[\blam(x_i)\T\Gamma\bD=0]$ for some non-zero $\bD$ determined by $\tau_j$. By \eqref{eq:indep} there is $\bu$ with $\Lambda_C\T\bu=\mathbf{0}$ and $\blam(x_i)\T\bu=1$. Choose $\bv$ with $\bv\T\bD=1$. The direction $\Delta=\bu\bv\T$ preserves $\Lambda_C\T\Gamma=K$ and leaves every $A$-factor unchanged; hence it lies in the linear part of $\mathcal{S}(\tau_j)$. The affine function $f(\Gamma)=\blam(x_i)\T\Gamma\bD$ satisfies $f(\Gamma+\Delta)=f(\Gamma)+1$. By \cref{lem:affine}, $\Pr[\mathsf{bad}_j\mid\tau_j]=1/q$.

\emph{(B) $A$-shift.} Suppose the real decision is $[\bc_i\T A_{U^\star,k^\star}\bfe(m)=y]$, where $U^\star$ is honest, $m\notin\mathcal{M}_{k^\star}$, and $y$ is determined by $\tau_j$, $\Gamma$ and the other signing matrices. Condition additionally on $\Gamma$ and on all $A_{U,k}\neq A_{U^\star,k^\star}$. By Step~2, $A_{U^\star,k^\star}$ remains uniform on its affine factor. If $\Gamma$ is invertible, then \eqref{eq:indep} gives $\bc_i=\Gamma\T\blam(x_i)\notin\operatorname{row}(K)$; hence some $\bg$ satisfies $K\bg=\mathbf{0}$ and $\bc_i\T\bg=1$. Since $|\mathcal{M}_{k^\star}|\le\psi$ and $m\notin\mathcal{M}_{k^\star}$, \cref{lem:vdm} gives $\bfe(m)\notin\spn(E_{U^\star,k^\star})$; hence some $\bh$ satisfies $E_{U^\star,k^\star}\T\bh=\mathbf{0}$ and $\bh\T\bfe(m)=1$. The direction $\bg\bh\T$ preserves both $KA=B$ and $AE=R$, and changes $\bc_i\T A\bfe(m)$ by $1$. By \cref{lem:affine}, $\Pr[\mathsf{bad}_j\mid\tau_j,\Gamma]=1/q$ whenever $\Gamma$ is invertible.

\emph{Case analysis.} Let query $j$ be non-acceptable, since otherwise $\mathsf{bad}_j$ cannot occur. Malformed objects are rejected in both games.
(i)~\emph{Individual, corrupted signer.} The object is invalid; hence $\bD=\balpha-A_{U,k}\bfe(m)\neq\mathbf{0}$ is known from $\tau_j$, and the real decision is $[\blam(x_i)\T\Gamma\bD=0]$. Shift (A) applies. (ii)~\emph{Individual, honest signer, $m\in\mathcal{M}_k$, $\balpha$ differs from the returned signature $\balpha_{k,m}$.} Set $\bD=\balpha-\balpha_{k,m}\neq\mathbf{0}$. Shift (A) applies. (iii)~\emph{Individual, honest signer, $m\notin\mathcal{M}_k$.} Shift (B) applies with $y=\bc_i\T\balpha$. (iv)~\emph{Aggregate in which every honest $U\in S$ has signed $(k_U,m)$.} Then $\sum_{U\in S}A_{U,k_U}\bfe(m)$ is known from $\tau_j$, and non-acceptability means $\bD:=\balpha_S-\sum_U A_{U,k_U}\bfe(m)\neq\mathbf{0}$. Shift (A) applies. (v)~\emph{Aggregate in which some honest $U^\star\in S$ has not signed $(k_{U^\star},m)$.} Shift (B) applies with $y=\bc_i\T\balpha_S-\sum_{U\ne U^\star}\bc_i\T A_{U,k_U}\bfe(m)$.
In every case $\Pr[\mathsf{bad}_j\wedge\Gamma\text{ invertible}]\le 1/q$. A uniform $(\omega+1)\times(\omega+1)$ matrix is singular with probability $1-\prod_{\ell=1}^{\omega+1}(1-q^{-\ell})\le\sum_{\ell\ge1}q^{-\ell}=1/(q-1)$. A union bound gives
\begin{align*}
\Pr[\Bad]&\le\Pr[\Gamma\text{ singular}]+\textstyle\sum_{j=1}^{Q}\Pr[\mathsf{bad}_j\wedge\Gamma\text{ invertible}]\\
&\le\frac{1}{q-1}+\frac{Q}{q}.\qedhere
\end{align*}
\end{proof}

\textbf{Posterior uniformity in detail.}\label{app:posterior} We make Step~2 of the proof of \cref{thm:main} explicit. Fix the adversary's random tape $\rho$ and the corrupted signers' matrices $\mathbf{a}=(A_{U,k})_{U\in\mathcal{U}_C}$. The honest secrets are $\mathbf{s}=(\Gamma,(A_{U,k})_{U\notin\mathcal{U}_C})$, uniform and independent of $(\rho,\mathbf{a})$. In $G_1$, the adversary's next query is a deterministic function $\mathsf{Q}_j(\rho,\tau_{j})$ of its tape and the transcript up to that point. Each answer is either a signature $A_{U,k}\bfe(m)$ or an ideal bit $\Ideal_j$, which is a function of $(\tau_j,\mathbf{a})$ alone.

By induction on $j$, the transcript $\tau_{j}$ is a deterministic function $\mathsf{T}_j(\rho,\mathbf{a},\mathbf{s})$. Moreover, for every value $\hat\tau$,
\[
\{\mathbf{s}:\mathsf{T}_j(\rho,\mathbf{a},\mathbf{s})=\hat\tau\}=\mathcal{S}(\hat\tau),
\]
the product of affine subspaces of Step~2, provided $\hat\tau$ is consistent with $(\rho,\mathbf{a})$. Indeed, $\hat\tau$ fixes the sequence of queries. The initial keys contribute the constraints $\Lambda_C\T\Gamma=K$ and $KA_{U,k}=B_{U,k}$. Signing answers contribute $A_{U,k}\bfe(m)=\balpha$. Ideal answers contribute no constraint on $\mathbf{s}$. Hence $\Pr[\mathbf{s}=\hat{\mathbf{s}}\mid\tau_j=\hat\tau,\rho,\mathbf{a}]=\ind[\hat{\mathbf{s}}\in\mathcal{S}(\hat\tau)]/|\mathcal{S}(\hat\tau)|$. Averaging over $(\rho,\mathbf{a})$ consistent with $\hat\tau$ preserves uniformity on $\mathcal{S}(\hat\tau)$, since $\mathcal{S}(\hat\tau)$ does not depend on them.

\begin{remark}[What the proof does \emph{not} need]
The proof uses no computational assumption, no random oracle, and no rewinding. It holds even if the adversary chooses its queries with unbounded computation. It is also robust to how $\bc_s$ is \emph{distributed} among corrupted verifiers: Step~3 only needs $\blam(x_i)\notin\spn\{\blam(x_s)\}$, which holds deterministically. For the variant in which each verifier samples a uniformly random $\bc_i$ itself (no anchors), the same proof goes through with a different singular event, $\bc_i\in\operatorname{row}(K)$, which has probability at most $1/q$ per query, giving $\Pr[\Bad]\le 2Q/q$.
\end{remark}

\begin{remark}[Several epochs]\label{rem:epochs}
\Cref{thm:main} concerns one anchor matrix. Over $E$ epochs with independently sampled anchors, where epoch $e$ answers $Q_e$ queries, the same proof applied per epoch and a union bound give $\Pr[\Bad]\le\sum_e Q_e/q+\sum_e p_{\mathrm{sing}}\le Q_{\mathrm{total}}/q+E/(q-1)$, where $p_{\mathrm{sing}}=1-\prod_{\ell=1}^{\omega+1}(1-q^{-\ell})$. Only epochs in which a decision is taken need to be counted. Further independent components, such as the scalar anchors of arbiter-verified keys (\cref{sec:av}), redistribution (\cref{lem:reshare}), MAC forgeries and probabilistic decoding (\cref{thm:tn}), contribute their own terms to a system-level bound.
\end{remark}

\begin{remark}[Invertible anchors]\label{rem:invertible}
If $\Gamma$ is sampled uniformly from the invertible matrices, the bound becomes $\Pr[\Bad]\le Q/q$, but the proof changes, because $\{\Gamma:\Lambda_C\T\Gamma=K,\ \det\Gamma\neq0\}$ is not an affine subspace and shift (A) may leave it. In the ideal game, fix $K$ with row space $W$ of dimension $c=|C|$ (it has full rank, since $\Gamma$ is invertible). By \eqref{eq:indep}, the honest functional $\bc_i\T=\blam(x_i)\T\Gamma$ is then uniform on $\F^{\omega+1}\setminus W$. For a transcript-determined $\bD\neq\mathbf{0}$, $\Pr[\bc_i\T\bD=0\mid K]=(q^{\omega}-|W\cap\bD^{\perp}|)/(q^{\omega+1}-q^{c})\le1/q$, because $|W\cap\bD^{\perp}|\ge q^{c-1}$. The fresh-message case keeps shift (B), conditional on $\Gamma$. The distributed protocol $\Pki$ of \cref{fig:dki} does not implement this key distribution: doing so needs a secure invertibility test (e.g., opening $\Gamma R$ for a random shared $R$ and checking its rank, which reveals nothing when both are invertible), with its own multiplications, retries and leakage analysis. We therefore keep uniform anchors and the $1/(q-1)$ term.
\end{remark}

\subsection{Consequences}
\begin{corollary}[Strong unforgeability and non-repudiation]\label{cor:uf}
Except with probability $Q/q+1/(q-1)$:
no honest verifier, and not the arbiter, accepts a signature of an honest signer that the signer did not output, including modified versions of genuine signatures; no honest verifier accepts an invalid signature of a dishonest signer.
Consequently, a signer cannot repudiate a signature that any honest verifier has accepted: $\Res$ upholds it.
\end{corollary}

\begin{corollary}[Unbounded-hop transferability]\label{cor:transfer}
Let a fixed signature or aggregate be forwarded through honest verifiers $V_{i_1},V_{i_2},\dots,V_{i_L}$, for \emph{any} $L$. If some $V_{i_\ell}$ rejects it although $V_{i_1}$ accepted it, then the acceptance by $V_{i_1}$ was a $\Bad$ event. Hence the probability of such a disagreement is at most $Q/q+1/(q-1)$, where $Q$ counts the verification decisions taken up to and including that first acceptance; the $L-1$ later verifications of the same object add nothing.
\end{corollary}
\begin{proof}
If $V_{i_1}$ accepts a non-acceptable object, then $\Bad$ occurs, which \cref{thm:main} bounds with the queries made so far. Otherwise the object is valid and, by perfect correctness, every honest verifier accepts it deterministically.
\end{proof}
Transferability is thus a structural property: once an object is valid, every honest verifier accepts it, at every hop. It is not a guarantee against unboundedly many \emph{fresh} verification attempts, which \cref{thm:main} charges $1/q$ each.
In threshold-based USS and QDS~\cite{Amiri2018,Wallden2015}, each hop lowers the acceptance level, which bounds the number of transfers. In \scheme{}, validity is a single global predicate, and there is nothing to degrade.

\begin{corollary}[Information-theoretic aggregate signatures]\label{cor:agg}
Except with probability $Q/q+1/(q-1)$, an honest verifier accepts an aggregate $\sigma_S$ on $m$ only if every honest signer in $S$ actually signed $m$ under the stated index, and $\balpha_S$ equals the sum of the corresponding valid signatures. Rogue-key attacks do not apply, because no participant chooses its own key.
\end{corollary}

\begin{corollary}[Simulation-based security in the Fitzi framework]\label{cor:fds}
Let $\omega\ge1$, and let $P_1\in\mathcal{U}$, $P_2=V_i$ and $P_3=V_j$, with keys issued by $\Fki$ independently of the input $x_1\in\F$. Consider the following non-interactive protocol. In the signing phase, $P_1$ sends $\sigma\gets\Sign(P_1,k,x_1)$ to $P_2$ over a private authenticated channel, which outputs $z_2=x_1$ if $\Ver_i(\sigma)=1$ and $\perp$ otherwise. In the transfer phase, $P_2$ forwards $\sigma$ to $P_3$, which outputs $z_3$ analogously. This protocol satisfies the correctness, unforgeability, transferability and secrecy conditions of the characterization in~\cite[Lemma~18.15]{TyagiWatanabe2023} with error $\varepsilon\le 1/q+1/(q-1)$. Under the hypotheses of that lemma (three parties, one static active corruption, pairwise private authenticated channels, classical inputs), it $2\varepsilon$-securely implements the three-party signature functionality $f_{\mathsf{DS}}$ of Fitzi~\emph{et al.}~\cite{Fitzi2002,FitziWolf2004,TyagiWatanabe2023}. The four game-based properties also hold against quantum adversaries that interact classically (\cref{rem:quantum}). A simulation-based statement against quantum adversaries additionally relies on the quantum version of the characterization in~\cite{SILMARILS}, whose hypotheses must be matched to this protocol; we state it only in that conditional form, and we make no claim for adversaries with entangled auxiliary inputs or for a particular composable quantum framework.
\end{corollary}
\begin{proof}
Correctness is perfect (\cref{prop:correct}).
\emph{Unforgeability} ($P_1,P_3$ honest, $P_2$ corrupt, whence $|C|=1\le\omega$): $z_3\notin\{x_1,\perp\}$ requires $P_3$ to accept a non-authentic signature. Since $Q=1$, this has probability at most $1/q+1/(q-1)$. \emph{Transferability} ($P_2,P_3$ honest, $P_1$ corrupt): $z_2\ne\perp$ and $z_3\ne z_2$ require that $P_2$ accept an invalid signature, because a valid one is accepted by $P_3$. By shift~(A), this has probability at most $1/q$. \emph{Secrecy} ($P_3$ corrupt): the signing message travels over the private channel from $P_1$ to $P_2$, and $P_3$ therefore receives nothing during signing; its view consists of its keys, which are independent of $x_1$, and $\varepsilon_{\mathrm{sec}}=0$. Without a private signing channel, a corrupt $P_3$ could read $x_1$ from $\sigma$, and this condition would fail.
The classical implication then follows from~\cite[Lemma~18.15]{TyagiWatanabe2023}; the quantum implication is conditional on the quantum version of that lemma in~\cite{SILMARILS}.
\end{proof}
\Cref{cor:fds} recovers, non-interactively, without broadcast, and for an unbounded number of transfers, what the three-party mode of \silm{} achieves with an interactive IC phase and a single transfer.

\begin{remark}[Quantum adversaries]\label{rem:quantum}
All keys and all messages exchanged with honest parties are classical. In the game of \cref{def:exp}, any quantum adversary therefore induces a (possibly inefficiently computable) randomized strategy for choosing its classical queries, and the proof of \cref{thm:main} holds for every such strategy. This is a statement about the game-based bound; simulation-based or composable quantum security needs the separate argument discussed after \cref{cor:fds}. The proof does not rely on any computational restriction on the adversary. Consequently, the information-theoretic component of the security analysis remains unchanged even if the adversary possesses quantum computational capabilities, and \scheme{} has no computational assumption for a quantum computer to break. The only exception is collision resistance when long messages are hashed (\cref{rem:hash}). Because hashed messages are reduced into $\F$, the generic quantum collision attack~\cite{BHT1998} needs about $q^{1/3}\approx2^{85}$ queries whatever the digest length; a larger margin uses the vector encoding of \cref{rem:hash} (\cref{sec:construction}, parameter selection).
\end{remark}

\subsection{Tightness}
\begin{proposition}[Tightness]\label{prop:tight}
(a)~\emph{The bound of \cref{thm:main} is tight up to lower-order terms.} An adversary without any corruption makes $Q\le q$ queries on a fixed unsigned message $m$, using $\balpha_\ell=\balpha_1+(\ell-1)\mathbf{w}$ for a uniform $\mathbf{w}$. It succeeds with probability at least $(Q/q)(1-1/q)(1-1/(q-1))$. (b)~\emph{The collusion threshold is tight.} If $\omega+1$ verifiers collude, they recover $\Gamma$ and can compute every honest verifier's response to every message, i.e., they forge universally with probability $1$. (c)~\emph{The capacity is tight.} Given $\psi+1$ signatures on distinct messages under one index, anyone recovers $A_{U,k}$ and forges on that index with probability $1$. (d)~\emph{The non-repudiation term is exact.} For a first query with an invalid $\balpha$, each individual honest verifier accepts with probability exactly $1/q$.
\end{proposition}
\begin{proof}
\emph{(a)} Let $\Gamma$ be invertible; then $\bc_i\neq\mathbf{0}$. Let $\theta=\bc_i\T A\bfe(m)$, which is uniform by shift~(B). If $\bc_i\T\mathbf{w}\neq0$, which happens with probability $1-1/q$, the $Q$ values $\bc_i\T\balpha_\ell$ are distinct. Each rejection excludes one candidate value of $\theta$; hence some $\balpha_\ell$ hits $\theta$ with probability $Q/q$.
\emph{(b)} By \cref{lem:vdm}, $\Lambda_C$ is invertible; hence $\Gamma=\Lambda_C^{-\mathsf T}K$ and every $\bc_i=\Gamma\T\blam(x_i)=K\T\Lambda_C^{-1}\blam(x_i)$ is a known combination of the rows of $K$. Write $\bc_i=K\T\boldsymbol{\mu}_i$ with $\boldsymbol{\mu}_i=\Lambda_C^{-1}\blam(x_i)$ public. For any $m$, the coalition picks any solution $\balpha$ of $K\balpha=B_{U,k}\bfe(m)$. The system is consistent because $A_{U,k}\bfe(m)$ solves it. Then $\bc_i\T\balpha=\boldsymbol{\mu}_i\T K\balpha=\boldsymbol{\mu}_i\T B_{U,k}\bfe(m)=\bb_{i,U,k}\T\bfe(m)$ for every $i$, and the forgery is accepted by everyone.
\emph{(c)} $E_{U,k}$ is an invertible Vandermonde matrix; hence $A_{U,k}=R_{U,k}E_{U,k}^{-1}$.
\emph{(d)} This is shift (A) with equality.
\end{proof}

Part~(a) establishes that the leading term $Q/q$ is tight; it does not claim that the whole bound $Q/q+1/(q-1)$ is attained for every parameter set. Our Monte-Carlo experiments (\cref{sec:eval:security}) are consistent with (a) for $Q\in\{1,4,16\}$, where success grows as $Q/q$, and with (b), (c) and (d); they support the analysis but do not replace it.

\section{Dealer-Free Distributed Key Issuance}\label{sec:dki}

We now realize $\Fki$ without a trusted dealer. The committee $\mathcal{P}$ evaluates the arithmetic circuit of \cref{fig:fki} using perfectly secure honest-majority MPC. Our contribution is to exhibit a circuit, enabled by virtual anchors, whose \emph{multiplicative} work is independent of the number $\nu$ of verifiers, and to count its cost: exact circuit counts, a semi-honest communication cost validated by simulation, and an estimate for the malicious setting.

\subsection{Building Blocks}
We use degree-$t$ Shamir sharings $\sh{s}$ over $\F$ among $n\ge 3t+1$ parties~\cite{Shamir1979}, and the following standard sub-protocols, each perfectly secure against $t<n/3$ malicious parties~\cite{BGW1988,AsharovLindell2017,BTH2008,CDN2015}:
$\mathsf{Rand}$, which outputs $\sh{r}$ for a uniform $r$, generated in batches by Vandermonde (hyper-invertible) extraction~\cite{DN2007,BTH2008}; $\mathsf{IP}$, which on input $\sh{\mathbf{a}},\sh{\mathbf{b}}$ outputs $\sh{\mathbf{a}\T\mathbf{b}}$. Each party multiplies and sums its local shares, obtaining a degree-$2t$ sharing, and one degree reduction follows. An inner product therefore costs one degree reduction, i.e., the communication of one multiplication; the local coordinate products cost no communication~\cite{DN2007,BTH2008}; $\mathsf{Open}_{\to P}$, which reconstructs a shared value towards a single party. All $n$ shares are sent, and the receiver decodes the degree-$t$ Reed--Solomon codeword with up to $t$ errors (e.g., Berlekamp--Welch~\cite{BerlekampWelch}).
Linear combinations with public coefficients are computed locally. These sub-protocols use broadcast only for complaints and for short fault-detection messages, amortized over large batches; for $t<n/3$ we realize it by Byzantine agreement over point-to-point channels~\cite{PSL1980}. Signing, aggregation and verification never use it.

\subsection{The Protocol}
\begin{figure}[!t]
\fbox{\begin{minipage}{0.96\columnwidth}\small
\textbf{Protocol} $\Pki$ (committee $\mathcal{P}$, $t<n/3$)\\[2pt]
\emph{Epoch set-up.}
\begin{enumerate}
\item $\sh{\Gamma}\gets\mathsf{Rand}^{(\omega+1)^2}$.
\item For each verifier $V_i$: locally compute $\sh{\bc_i}=\sh{\Gamma}\T\blam(x_i)$ and run $\mathsf{Open}_{\to V_i}$ on its $\omega+1$ entries. Keep $\sh{\bc_0}=\sh{\Gamma}\T\blam(0)$, i.e., the first row of $\sh{\Gamma}$, as the arbiter key.
\end{enumerate}
\emph{Per signer-index $(U,k)$} (batched and pipelined):
\begin{enumerate}\setcounter{enumi}{2}
\item $\sh{A}\gets\mathsf{Rand}^{(\omega+1)(\psi+1)}$.
\item $\sh{M}\gets\sh{\Gamma}\sh{A}$, using $(\omega+1)(\psi+1)$ calls to $\mathsf{IP}$ of length $\omega+1$.
\item $\mathsf{Open}_{\to U}$ on every entry of $\sh{A}$.
\item For each $V_i$: locally compute $\sh{\bb_{i}}=\sh{M}\T\blam(x_i)$, by \eqref{eq:keypoly}, and run $\mathsf{Open}_{\to V_i}$ on its $\psi+1$ entries. Keep $\sh{\bb_0}=\sh{M}\T\blam(0)$.
\end{enumerate}
\emph{Resolution.} On $\sigma$: $\sh{\delta}=\sh{\bc_0}\T\balpha-\sh{\bb_{0,U,k}}\T\bfe(m)$ (local). Take a fresh preprocessed uniform non-zero $\sh{\rho}$ (\cref{sec:construction}), compute $\sh{\delta\rho}$ with one multiplication, send it to all parties point-to-point, decode it robustly, and output $[\delta\rho=0]$. Robust decoding gives every honest party the same verdict; hence no broadcast is needed.
\end{minipage}}
\caption{Distributed key issuance with virtual anchors. Only step 4 and resolution use multiplications; verifier keys (steps 2 and 6) are local linear maps followed by private openings.}
\label{fig:dki}
\end{figure}

\Cref{fig:dki} gives $\Pki$. It computes exactly the outputs of $\Fki$: $\bc_i=\Gamma\T\blam(x_i)$, $A_{U,k}$, and $\bb_{i,U,k}=M_{U,k}\T\blam(x_i)=A_{U,k}\T\Gamma\T\blam(x_i)=A_{U,k}\T\bc_i$. $\Gamma$, $M_{U,k}$ and the arbiter's key are never opened.

\begin{theorem}[Security of $\Pki$]\label{thm:dki}
Let $\Pki$ be instantiated with perfectly secure malicious sub-protocols for $\mathsf{Rand}$, $\mathsf{IP}$ and $\mathsf{Open}$ (e.g., BGW~\cite{BGW1988,AsharovLindell2017}, or the linear-communication protocols of~\cite{BTH2008} with fault detection and player elimination). Then $\Pki$ securely computes $\Fki$ with perfect security and guaranteed output delivery against a static malicious adversary that corrupts at most $t<n/3$ committee members, together with any set of signers and verifiers, given synchronous private channels. Consequently, \cref{thm:main} and its corollaries hold in the real world, with $\Pki$ in place of $\Fki$, with the same bounds.
\end{theorem}
\begin{proof}[Proof sketch]
$\Fki$ is a randomized arithmetic circuit of multiplicative depth one, with outputs to designated parties. Corrupted verifiers and signers have no inputs. By the BGW theorem with its full proof~\cite{BGW1988,AsharovLindell2017}, every such circuit has a perfectly secure protocol for $t<n/3$ with guaranteed output. The sub-protocols above are standard perfectly secure instantiations whose composition is perfectly secure~\cite{BTH2008,CDN2015}. Robust opening succeeds because a degree-$t$ codeword of length $n\ge3t+1$ corrects $t$ errors.

Privacy of the unopened values follows from the simulator. Corrupted parties hold at most $t$ shares of each degree-$t$ sharing; hence their view of $\Gamma$, $M_{U,k}$, $\bc_0$ and $\bb_{0,\cdot}$ is simulatable from their outputs alone. The opened value $\delta\rho$ is $0$ or uniform on $\F^{*}$, hence simulatable from the verdict.

Perfectly secure protocols with straight-line black-box simulators remain secure under composition~\cite{KLR2010,Canetti2000}. The joint distribution of the adversary's view and all honest outputs is therefore \emph{identical} in the real and ideal executions, and any event, $\Bad$ in particular, has the same probability in both.
\end{proof}

\subsection{Communication Cost}
Let $d$ be the number of random sharings extracted from $n$ dealt sharings: $d=n-t$ with plain Vandermonde extraction (semi-honest dealers), and $d=n-2t$ with hyper-invertible extraction, which also verifies the dealt sharings against malicious dealers~\cite{BTH2008}. Let $r_c=n(n-1)/d$ be the amortized cost, in field elements sent in total, of one $\mathsf{Rand}$ output. Let $\mu_c=2(n-1)+2n(n-1)/d$ be the cost of one degree reduction with double sharings and a king~\cite{DN2007}, and let $n$ (an upper bound) be the cost of one $\mathsf{Open}$. With $E=(\omega+1)(\psi+1)$, the communication per signer-index is
\begin{equation}\label{eq:dkicost}
C_{\mathrm{idx}}=\underbrace{E\,r_c}_{\text{random }A}+\underbrace{E\,\mu_c}_{\Gamma A}+\underbrace{E\,n}_{A\to U}+\underbrace{\nu(\psi+1)\,n}_{\bb_i\to V_i}
\end{equation}
field elements. The epoch set-up costs a further $(\omega+1)^2r_c+\nu(\omega+1)n$. Our simulator implements the semi-honest instantiation ($d=n-t$, king-based degree reduction), with Berlekamp--Welch decoding of all outputs against $t$ parties that send corrupted shares. It agrees with \eqref{eq:dkicost} to within $0.4\%$, e.g., exactly $908$ elements per index for $n=7$, $\omega=2$, $\psi=4$, and it confirms that all issued keys satisfy $\bb_i=A\T\bc_i$. It does \emph{not} model malicious dealers or a malicious king. We therefore separate four statements. (i)~The circuit counts ($(\omega+1)(\psi+1)$ multiplications per index, none per verifier) are exact. (ii)~\Cref{eq:dkicost} with $d=n-t$ is the semi-honest amortized communication, validated by simulation. (iii)~\Cref{eq:dkicost} with $d=n-2t$ is an \emph{estimate} of the malicious communication: perfect malicious security requires the fault detection, complaints and player elimination of~\cite{BTH2008} and the emulation of their broadcasts, which keep the communication linear but add terms that \eqref{eq:dkicost} omits, as do batch rounding, set-up and failure paths. (iv)~\Cref{thm:dki} establishes that a perfectly secure malicious realization exists. \Cref{sec:eval} reports (ii) and (iii) and labels (iii) as an estimate.

\textbf{Derivation.}\label{app:cost} $\mathsf{Rand}$ is batched: every party deals one degree-$t$ sharing, sending $n-1$ elements, and the $n$ dealt sharings are multiplied by a $d\times n$ extraction matrix. With $d=n-t$ and a Vandermonde matrix, this yields $n-t$ uniformly random sharings if dealers follow the protocol~\cite{DN2007}. With $d=n-2t$ and a hyper-invertible matrix, the remaining $t$ outputs are opened to verify consistency against malicious dealers~\cite{BTH2008}. Hence $r_c=n(n-1)/d$. A degree reduction consumes one double sharing $(\sh{r}_t,\sh{r}_{2t})$, generated the same way at $2n(n-1)/(n-t)$ elements, and one masked opening to a king and back, at $2(n-1)$ elements. $\mathsf{Open}_{\to P}$ costs $n$ elements. Summing over steps 3--6 of \cref{fig:dki} gives \eqref{eq:dkicost}. The residual between simulation and \eqref{eq:dkicost} comes from rounding up partially consumed batches.

\begin{remark}[Why virtual anchors matter]
The naive alternative lets each verifier input a secret random $\bc_i$ and computes $\bc_i\T A$ by MPC. This costs $\nu(\psi+1)$ multiplications per index instead of $(\omega+1)(\psi+1)$, a factor $\nu/(\omega+1)=(3t+1)/(t+1)$ more multiplications when $\nu=n=3t+1$ and $\omega=t$, which approaches $3$ for large committees. Our simulators measure a total per-index cost $1.5\times$, $1.65\times$ and $1.72\times$ higher for $n=4,7,10$ ($\psi=4$); the gap widens with $n$ because the multiplication term dominates.
\end{remark}

\begin{remark}[Further optimizations]
The workload is a large SIMD batch of identical depth-one circuits, the ideal case for packed secret sharing~\cite{FranklinYung1992} and for constant-online-communication honest-majority protocols~\cite{TurboPack2022}. If the committee accepts a slightly lower corruption threshold, e.g.\ $n\ge4t+1$, packing reduces the dominant opening term $\nu(\psi+1)n$ by a factor of $\Theta(n/t)$. We leave engineering these variants to future work.
\end{remark}

\section{Integration into Permissioned BFT Ledgers}\label{sec:ledger}

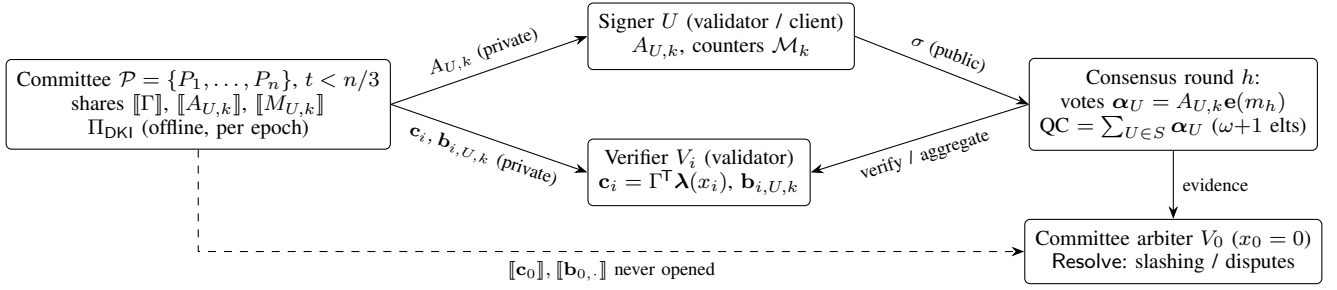
\begin{figure*}[!t]
\centering
\begin{tikzpicture}[font=\footnotesize, >=Stealth,
  box/.style={draw, rounded corners=2pt, align=center, inner sep=4pt, fill=white},
  comm/.style={draw, rounded corners=3pt, dashed, inner sep=6pt},
  lbl/.style={font=\scriptsize, align=center}]
\node[box, minimum width=3.6cm] (dki) {Committee $\mathcal{P}=\{P_1,\dots,P_n\}$, $t<n/3$\\ shares $\sh{\Gamma}$, $\sh{A_{U,k}}$, $\sh{M_{U,k}}$\\ $\Pki$ (offline, per epoch)};
\node[box, right=2.6cm of dki, yshift=0.9cm] (sig) {Signer $U$ (validator / client)\\ $A_{U,k}$, counters $\mathcal{M}_k$};
\node[box, right=2.6cm of dki, yshift=-0.9cm] (ver) {Verifier $V_i$ (validator)\\ $\bc_i=\Gamma\T\blam(x_i)$, $\bb_{i,U,k}$};
\node[box, right=2.3cm of sig, yshift=-0.9cm, minimum width=3.3cm] (cons) {Consensus round $h$:\\ votes $\balpha_U=A_{U,k}\bfe(m_h)$\\ QC $=\sum_{U\in S}\balpha_U$ ($\omega{+}1$ elts)};
\node[box, below=0.9cm of cons] (arb) {Committee arbiter $V_0$ ($x_0=0$)\\ $\Res$: slashing / disputes};
\draw[->] (dki.east) -- node[lbl, above, sloped]{$A_{U,k}$ (private)} (sig.west);
\draw[->] (dki.east) -- node[lbl, below, sloped]{$\bc_i,\bb_{i,U,k}$ (private)} (ver.west);
\draw[->] (sig.east) -- node[lbl, above, sloped]{$\sigma$ (public)} (cons.west);
\draw[->] (cons.west) -- node[lbl, below, sloped]{verify / aggregate} (ver.east);
\draw[->] (cons.south) -- node[lbl, right]{evidence} (arb.north);
\draw[->, dashed] (dki.south) |- node[lbl, below, pos=0.75]{$\sh{\bc_0},\sh{\bb_{0,\cdot}}$ never opened} (arb.west);
\end{tikzpicture}
\caption{\scheme{} in a permissioned BFT ledger. Key issuance is offline and private; signing, aggregation and verification are online, non-interactive and public; dispute resolution reuses the shared state.}
\label{fig:arch}
\end{figure*}

\Cref{fig:arch} shows the deployment. We consider a leader-based BFT protocol in the style of PBFT/HotStuff~\cite{PBFT2002,HotStuff2019,FastHotStuff2024}: $n=3t+1$ validators, a leader per view, and QCs formed from $2t+1$ votes. We first describe the \emph{individual-vote} (IV) mode, in which all validators are verifiers ($\mathcal{V}=\mathcal{P}$) and $\omega=t$; \cref{sec:modes} compares it with the modes of \cref{sec:scale}.

\subsection{Epochs and Key Refresh}
Time is divided into epochs, delimited by validator-set reconfiguration or by key exhaustion. During epoch $e-1$, the committee runs $\Pki$ for epoch $e$ in the background, sized for the expected number of votes and transactions plus headroom. Signing keys and verifier keys of past epochs can be erased once their epoch is final. The arbiter's shares $\sh{\bc_0},\sh{\bb_{0,\cdot}}$ are kept for a dispute window and handed over to new committees by redistribution (\cref{sec:reshare}). Since every epoch uses fresh $\Gamma$ and $A$, a verifier compromised in one epoch learns nothing about the next. Within an epoch, however, the collusion bound $\omega$ counts every verifier credential for that epoch's $\Gamma$ that is ever corrupted, also at different times (\cref{prop:enrol}). The per-epoch cost is quantified in \cref{sec:eval}.

\subsection{Message Format, Slots and Index Streams}
Every consensus message has an explicit \emph{slot} $s=(e,h,r,\phi)$: epoch, height, round and phase (e.g., prepare, commit, proposal). A vote is the structured message
\[
M=(\textsf{VOTE},\,s,\,\mathrm{digest}(B)),\qquad m=H(M).
\]
Every honest validator signs at most one message per slot. Each validator keeps a separate \emph{index stream} per message type. Its $j$-th signature in a stream uses index $k=\lfloor j/\psi\rfloor$; hence an honest validator never exceeds $\psi$ distinct messages per index. Extra rounds caused by view changes consume more capacity. Epochs are therefore provisioned with headroom, and the committee runs a small top-up issuance if a stream nears exhaustion. Exhaustion can only delay a validator (liveness); it never affects safety.

\subsection{Votes and Aggregated Quorum Certificates}
The leader verifies each incoming vote individually, at a cost of $\omega+\psi+2$ multiplications, and then forms
\[
\mathsf{QC}=\big(M,\ \mathrm{bitmap}(S),\ (k_U)_{U\in S},\ \textstyle\sum_{U\in S}\balpha_U\big)
\]
with $|S|=2t+1$.
The leader includes this QC in the next proposal. Each replica runs $\AggVer$ in $\omega+1+|S|(\psi+1)$ multiplications. By \cref{cor:agg}, an accepted QC certifies that every honest member of $S$ voted for $M$, except with probability $Q/q+1/(q-1)$ (plus collision resistance of $H$). Quorum intersection therefore carries over unchanged, as do safety and liveness, which in such protocols rely only on unforgeability and transferability of votes~\cite{PBFT2002,HotStuff2019}. Since the QC is forwarded and re-verified by every replica, unbounded transferability (\cref{cor:transfer}) is exactly the property consensus requires. The explicit index list $(k_U)$ is needed because stream counters advance independently at each validator, e.g., across view changes. We encode each index in $4$ bytes; all QC sizes and wire sizes reported in \cref{sec:eval}, and those of the prototype, include this list ($4|S|$ bytes per QC). $\AggVer$ rejects a QC whose signer identities are not distinct (\cref{sec:construction}).

\subsection{Transactions and Relays}
In IV mode, a client $U$ obtains $N$ indices per epoch. It signs a transaction once, and the transaction is gossiped through any number of untrusted relays or mempools, which need no keys. Every validator verifies it with its own key. This resolves the relay problem that motivated the three-party mode of \silm~\cite{SILMARILS}, without interaction and without bounding the number of hops. For large committees, clients use arbiter-verified keys instead (\cref{sec:av}): their signatures are $32$ bytes and are checked once per block by the committee.

\subsection{Accountability and Slashing}
Forensics~\cite{Forensics2021,Polygraph2021} and slashing~\cite{Casper2017} punish a validator that signs two conflicting messages for the same slot.

\begin{definition}[Equivocation evidence]
Evidence against $U$ consists of two individual signatures $\sigma=(U,k,m,\balpha)$ and $\sigma'=(U,k',m',\balpha')$ together with preimages $M\neq M'$ such that $m=H(M)$, $m'=H(M')$, $M$ and $M'$ carry the same slot, and $\Res(\sigma)=\Res(\sigma')=1$.
\end{definition}

\begin{theorem}[Accountability]\label{thm:account}
(i)~If $U$ equivocates, its evidence is accepted with probability $1$. (ii)~For an honest $U$, the probability that \emph{any} evidence against $U$ is accepted is at most $Q/q+1/(q-1)+\mathsf{Adv}^{\mathrm{CR}}_H$.
With the purely information-theoretic encoding of \cref{rem:hash}, the last term vanishes.
\end{theorem}
\begin{proof}
(i) follows from perfect correctness of $\Res$. For (ii), an honest $U$ signs at most one message per slot. Unless $H(M)$ or $H(M')$ collides with a message that $U$ signed, at least one of $\sigma,\sigma'$ is non-authentic, and its acceptance by the arbiter is the event $\Bad$ of \cref{thm:main}.
\end{proof}

If the two conflicting messages are signed under the same index and push it beyond $\psi$ messages (e.g., $\psi=1$), equivocation reveals $A_{U,k}$ (\cref{prop:tight}(c)), and anyone can then forge on that index. Evidence nevertheless stays sound: an honest signer never produces the two signatures that reveal its key, and the leak is confined to the guilty signer's single index. Aggregated QCs certify the honest members of $S$. For per-validator attribution of corrupted members, which forensics needs, replicas retain the individual votes, or the leader forwards them on demand.

\subsection{Light Clients and External Auditors}
In IV mode, external parties do not hold verifier keys. They can query $\Res$, rely on a QC after enrolling as verifiers (\cref{prop:enrol}), in which case they count towards $\omega=t$, or use the TN certificates of \cref{sec:tn}, which are designed for them: constant size, a few microseconds to verify, and a collusion bound $\omega$ chosen independently of $n$.

\subsection{Choosing a Mode}\label{sec:modes}
\begin{table}[!t]
\centering
\caption{Deployment modes of \scheme{} ($n=3t+1$).}
\label{tab:modes}
\scriptsize
\setlength{\tabcolsep}{1.6pt}
\begin{tabular}{@{}llll@{}}
\toprule
 & \textbf{IV} (\S\ref{sec:ledger}) & \textbf{TN} (\S\ref{sec:tn}) & \textbf{AV} (\S\ref{sec:av})\\
\midrule
Signer / verifiers & validator / validators & committee / external & client / arbiter\\
Collusion bound $\omega$ & $t$ & independent of $n$ & $0$\\
Object size & $32(t{+}1){+}4|S|$\,B${}^{\dagger}$ & $32(\omega{+}1)$\,B & $32$\,B\\
Signer state & counter & slot number & counter\\
Attribution & per validator & committee & per signer\\
Issuance / val. / block & $\Theta(n^2)$ & $O(\omega+\nu)$ & $O(1)$ per sig.\\
Consensus auth. & signatures & pairwise MACs & --\\
Suited for & $n\lesssim30$ & to $n=1000$ & clients, slashing\\
\bottomrule
\multicolumn{4}{@{}l@{}}{$^{\dagger}$Plus an $n$-bit signer bitmap; $|S|$ signer indices of $4$\,B.}
\end{tabular}
\end{table}
\Cref{tab:modes} summarizes the three modes. They share the construction, the issuance protocol and \cref{thm:main}, and they can be combined. A consortium of a few dozen validators can run IV mode and obtain individually attributable, aggregatable votes. A large committee runs consensus over pairwise MACs~\cite{PBFT2002}, information-theoretic with Wegman--Carter pads, issues TN certificates to external verifiers, gives clients arbiter-verified keys, and, if slashing is required, adds arbiter-verified commit votes (accountable TN, \cref{sec:av}).

\section{Scaling to Large Committees}\label{sec:scale}

\subsection{Where the Quadratic Cost Comes From}\label{sec:scale:why}
In the deployment of \cref{sec:ledger}, three choices interact. Every validator is a signer, every validator is a verifier, and $\omega=t$ because corrupted validators hold verifier keys. Consequently, per block, $n$ signer-indices are consumed, each needs a key for each of the $n$ verifiers, and every key is opened by all $n$ committee members. The resulting $\Theta(n^2)$ field elements per validator per block, together with signatures of $\omega+1=t+1$ elements, are what limit the individual-vote mode to committees of a few dozen validators (\cref{sec:eval}).

None of the three choices is forced by the scheme. This section removes them with three mechanisms that reuse the construction and the proofs unchanged:
(i)~\emph{Threshold notarization} (TN, \cref{sec:tn}). The committee itself is the signer. Its signing matrices are never opened, a certificate is a single $(\omega+1)$-vector, and the verifiers are external parties (auditors, light clients, bridges, other committees) whose collusion bound $\omega$ is independent of $n$. (ii)~\emph{Arbiter-verified keys} (\cref{sec:av}). With $\omega=0$, the only verifier is the committee arbiter. Signatures are a single field element, issuance needs no verifier keys, and the committee checks a whole block of signatures with one online multiplication. This serves clients and, optionally, validator accountability. (iii)~\emph{Open verifier sets and evolving committees} (\cref{sec:enrol,sec:reshare}). Verifiers can be enrolled at any time, also retroactively, by local linear maps, and the shared state can be handed over to a new committee of a different size, which yields membership changes and proactive refresh.

\subsection{Threshold Notarization}\label{sec:tn}
\textbf{Setting.} Committee members $P_1,\dots,P_n$, $n\ge3t+1$, have public, pairwise distinct, non-zero share points $\eta_1,\dots,\eta_n$. External verifiers $V_1,\dots,V_\nu$ have identifiers $x_i$ as in \cref{sec:model}, and at most $\omega$ of them are corrupted. Validators do not hold verifier keys. A validator that wants to check certificates can be enrolled as a verifier, and it then counts towards $\omega$. The committee runs any BFT protocol $\Pi_{\mathsf{BFT}}$ over authenticated point-to-point channels that decides one value per \emph{certified slot} $s=0,1,2,\dots$ (e.g., per height) and satisfies \emph{agreement}: no two honest members decide different values for the same slot. We do not construct a new agreement protocol: \cref{thm:tn} is conditional on agreement of $\Pi_{\mathsf{BFT}}$. A protocol whose safety with MAC authenticators is proved for $t<n/3$, without signatures, is PBFT in the version of Castro and Liskov~\cite{PBFT2002}, including its view change. We instantiate its MACs information-theoretically in the Wegman--Carter style~\cite{WegmanCarter1981,Stinson1994}, with a polynomial hash so that no digest is needed: a message is encoded injectively, with a length prefix, as a vector $\boldsymbol{\mu}=(\mu_1,\dots,\mu_\ell)\in\F^{\ell}$, $\ell\le L_{\max}$, that binds (sender, receiver, epoch, counter, message type, full message), and its tag is $\tau=\kappa_2+\sum_{i=1}^{\ell}\kappa_1^{\,i}\mu_i$. Here $\kappa_1$ is a pre-shared key of the ordered pair (sender, receiver) for the epoch, and $\kappa_2$ is a fresh pre-shared $32$-byte pad for each message, indexed by a per-pair message counter that covers every message type, including view-change, new-view and checkpoint messages. Each pad authenticates at most one message, and a receiver accepts at most one message per counter. A failed check does \emph{not} consume the pad; hence injected candidates cannot prevent a genuine, delivered message from being accepted. Because every honest tag is masked by a fresh pad, $\kappa_1$ is uniform given the honest tags alone; rejections then exclude candidate values. A check of a message other than the one authenticated under that counter succeeds only if $\kappa_1$ is a root of a non-zero polynomial of degree at most $L_{\max}$ (or, for a counter not yet used, if the adversary hits the uniform pad), and each failed check excludes at most $L_{\max}$ candidates. Hence, if honest members perform at most $N_{\mathrm{chk}}$ tag checks in the epoch, the probability that some check accepts a forged message is
\[
\varepsilon_{\mathsf{MAC}}\le\min\Big\{1,\ \frac{N_{\mathrm{chk}}L_{\max}}{q-N_{\mathrm{chk}}L_{\max}}\Big\}\quad(N_{\mathrm{chk}}L_{\max}<q),
\]
and $\varepsilon_{\mathsf{MAC}}\le1$ otherwise. $N_{\mathrm{chk}}$ is finite because a receiver processes at most its link capacity during an epoch; with $q\approx2^{255}$ the bound is negligible for any physical traffic. Fixing $N_{\mathrm{chk}}$ at the link capacity makes rekeying unnecessary within an epoch; $\kappa_1$ is fresh in every epoch, and a lifetime budget sums $\varepsilon_{\mathsf{MAC}}$ over pairs and epochs. Agreement may still rely on hashing elsewhere: PBFT agrees on request digests~\cite{PBFT2002}, which adds the collision advantage $\mathsf{Adv}^{\mathrm{CR}}_{H_{\mathsf{BFT}}}$ unless replicas agree on full values or on an injective encoding. We write $\varepsilon_{\mathsf{agr}}\le\varepsilon_{\mathsf{MAC}}+\mathsf{Adv}^{\mathrm{CR}}_{H_{\mathsf{BFT}}}+\varepsilon_{\mathsf{BFT}}$, where $\varepsilon_{\mathsf{BFT}}$ is the failure probability of the agreement protocol itself with ideal authenticated channels ($0$ for PBFT). In the normal case, each receiver needs only its own entry of a PBFT authenticator, which is the one-tag-per-delivered-message traffic we count in \cref{sec:eval:scale}; view-change and checkpoint traffic are outside that count. Slot $s$ uses the committee index $k(s)=\lfloor s/\psi\rfloor$, a fixed public map that crashes, recoveries and view changes do not alter. The value $v$ decided in slot $s$ is certified on the message $m_s(v)=H(\textsf{CERT},e,s,v)$, which separates certificates from all other message types.

\begin{figure}[!t]
\fbox{\begin{minipage}{0.96\columnwidth}\small
\textbf{Protocol} $\Pi_{\mathsf{TN}}$ (committee $\mathcal{P}$, $t<n/3$)\\[2pt]
\emph{Issuance} (offline, as in $\Pki$ without the opening to a signer).
\begin{enumerate}
\item $\sh{\Gamma}\gets\mathsf{Rand}^{(\omega+1)^2}$; open $\bc_i=\Gamma\T\blam(x_i)$ to each $V_i$.
\item Per committee index $k$: $\sh{A_k}\gets\mathsf{Rand}^{E}$, $\sh{M_k}\gets\sh{\Gamma}\sh{A_k}$, and open $\bb_{i,k}=M_k\T\blam(x_i)$ to each $V_i$. $A_k$ and $M_k$ are never opened. Member $P_j$ holds $A_k(\eta_j)$, where $A_k(y)$ is the degree-$t$ sharing polynomial with $A_k(0)=A_k$.
\end{enumerate}
\emph{Release rule.} When $P_j$ decides $v$ in slot $s$, and only then, it sends the share $\balpha^{(j)}=A_{k(s)}(\eta_j)\,\bfe(m_s(v))$ to the combiners. It releases at most once per slot.\\[2pt]
\emph{Combine.} Given shares for $(s,v)$:
(i) \emph{optimistic}: interpolate $\balpha$ at $0$ from $t+1$ shares of non-suspected members and, if the combiner is an enrolled verifier, check $\balpha$ with its own key;
(ii) \emph{otherwise}, or if the check fails: decode the degree-$t$ Reed--Solomon codeword. The \emph{fast decoder} freezes the set of shares received, then draws a fresh uniform $\boldsymbol{\rho}$, decodes the scalar word $(\boldsymbol{\rho}\T\balpha^{(j)})_j$ (Gao or Berlekamp--Welch), marks members whose scalars disagree as suspects, and interpolates every coordinate from $t+1$ non-suspects. A retry with additional shares freezes a new set and draws a new $\boldsymbol{\rho}$. The \emph{exact decoder} decodes each coordinate separately. Output the certificate $(s,v,\balpha)$.\\[2pt]
\emph{Verify.} $V_i$ accepts iff $\bc_i\T\balpha=\bb_{i,k(s)}\T\bfe(m_s(v))$.
\end{minipage}}
\caption{Threshold notarization. A certificate is $\omega+1$ field elements for every $n$; each member sends one share per certified slot.}
\label{fig:tn}
\end{figure}

\Cref{fig:tn} gives the protocol. A certificate is an ordinary \scheme{} signature of the \emph{committee}, viewed as a single signer whose signing matrices exist only in shared form. Its validity is the same global predicate as before; certificates are therefore transferable without bound and can be resolved by the arbiter. Shares are linear in the shared signing matrix: each member's work is one matrix--vector product, exactly the cost of a signature.

\begin{lemma}[Share simulation]\label{lem:sharesim}
Let $T$ be a set of $t$ members, let $A(y)$ be a uniformly random degree-$t$ matrix polynomial with $A(0)=A$, and let $m_1,\dots,m_L\in\F$. Then $(A(\eta_j))_{j\in T}$ is uniform and independent of $A$, and for every $\ell$ the honest shares $(A(\eta_j)\bfe(m_\ell))_{j\notin T}$ are a fixed linear function of $\balpha_\ell=A\bfe(m_\ell)$ and of $(A(\eta_j)\bfe(m_\ell))_{j\in T}$.
\end{lemma}
\begin{proof}
The first claim is the perfect privacy of Shamir sharing, applied entrywise. For the second, $y\mapsto A(y)\bfe(m_\ell)$ is a vector polynomial of degree at most $t$ whose values at the $t+1$ distinct points $\{0\}\cup\{\eta_j:j\in T\}$ are $\balpha_\ell$ and the corrupted shares. Lagrange interpolation at these points determines its value at every $\eta_j$.
\end{proof}

\begin{theorem}[Security of threshold notarization]\label{thm:tn}
Let $\Pi_{\mathsf{BFT}}$ satisfy agreement, and let $\Pi_{\mathsf{TN}}$ be instantiated with the sub-protocols of \cref{thm:dki}. Consider any computationally unbounded adversary that corrupts at most $t$ members and at most $\omega$ verifiers, and makes at most $Q$ verification and resolution queries.
(i)~\emph{Soundness.} The probability that some honest verifier, or the arbiter, accepts a certificate $(s,v,\balpha)$ although no honest member decided $v$ in slot $s$ is at most $Q/q+1/(q-1)+\varepsilon_{\mathsf{agr}}+\mathsf{Adv}^{\mathrm{CR}}_H$. The last term is the collision advantage of the adversary's strategy against the certificate hash $H$; it vanishes with the injective encoding of \cref{rem:hash}, which here means one distinct encoded message per index, one index per slot, and encoding dimension $L+1$. Hashing used inside agreement is accounted separately, in $\varepsilon_{\mathsf{agr}}$. (ii)~\emph{Uniqueness.} Except with the same probability, no two certificates for the same slot and different values are accepted by honest verifiers. (iii)~\emph{Robustness.} If all honest members decide $v$ in slot $s$, then every combiner that receives the shares of the honest members outputs $\balpha=A_{k(s)}\bfe(m_s(v))$, whatever the corrupted members send: with probability $1$ if it runs the exact decoder directly, and with probability at least $1-t/q$ for each invocation of the fast decoder on a frozen set of shares ($1-Rt/q$ over $R$ invocations). The optimistic path of \cref{fig:tn} returns early when the combiner's own check accepts an interpolated vector; it can return a wrong vector only if that check accepts an invalid object, which is a $\Bad$ event of \cref{thm:main} for the combiner's checks (probability $\varepsilon_{\mathsf{opt}}$). The optimistic protocol is therefore correct except with probability $\varepsilon_{\mathsf{opt}}$ with exact fallback, and $\varepsilon_{\mathsf{opt}}+Rt/q$ with $R$ fast-decoder calls. In a synchronous round, missing shares are erasures, and decoding succeeds from the at most $n$ shares received. Asynchronously, the combiner uses online error correction~\cite{BCG1993}: with $r\ge2t+1$ shares received, it decodes up to $\lfloor(r-t-1)/2\rfloor$ errors and outputs only if the decoded polynomial agrees with at least $2t+1$ received shares, and otherwise waits for another share; it terminates once $2t+1+e$ shares have arrived, where $e\le t$ is the number of incorrect shares among them. The optimistic checks count towards $Q$. A combiner that is not an enrolled verifier may output a wrong vector with the probability above; honest verifiers then reject it unless $\Bad$ occurs. (iv)~\emph{Capacity and transfer.} Conditioned on agreement, honest members release shares on at most $\psi$ distinct messages per index; hence the attack of \cref{prop:tight}(c) cannot be mounted from honest releases. Shares released by corrupted members are functions of their own shares and add nothing (\cref{lem:sharesim}). Agreement fails with probability at most $\varepsilon_{\mathsf{agr}}$, which (i) already accounts for. Accepted certificates are accepted by every honest verifier, after any number of hops (\cref{cor:transfer}).
\end{theorem}
\begin{proof}
\emph{(iv)}, capacity: by agreement, all honest members that release in slot $s$ release on the same value; each index serves $\psi$ slots; hence the honest members release shares on at most $\psi$ distinct messages per index, unless agreement fails, which has probability $\varepsilon_{\mathsf{agr}}$.

\emph{(i).} Let $\mathcal{F}_{\mathsf{TN}}$ be the reactive functionality that samples $\Gamma$ and the $A_k$, keeps them, outputs the verifier keys of $\Fki$ (also on later enrollment requests, \cref{prop:enrol}), and, when an honest member $P_j$ releases on $(k,m)$, hands $A_k(\eta_j)\bfe(m)$ to the adversary. Throughout $\Pi_{\mathsf{TN}}$, all state is kept as degree-$t$ sharings, and every step is $\mathsf{Rand}$, $\mathsf{IP}$, a local linear map, a robust opening, or the release of a local linear map of a member's own shares. The perfectly secure reactive composition used for \cref{thm:dki}~\cite{BGW1988,AsharovLindell2017,KLR2010} therefore shows that $\Pi_{\mathsf{TN}}$ securely realizes $\mathcal{F}_{\mathsf{TN}}$; in the ideal world, the corrupted members' shares of each $A_k$ are uniform and independent of all secrets. Given an adversary $\mathcal{B}$ against $\Pi_{\mathsf{TN}}$, we build an adversary $\mathcal{A}$ against the experiment of \cref{def:exp} with a single honest signer $U^\star$ (the committee) whose index-$k$ signing matrix is $A_k$. Giving $\mathcal{B}$ the shares of further honest members only increases its information; we therefore assume $|T|=t$. $\mathcal{A}$ corrupts the same verifiers as $\mathcal{B}$, samples the corrupted members' shares uniformly, and runs $\mathcal{B}$ together with all honest members of $\Pi_{\mathsf{BFT}}$; consensus uses no secret of \scheme{}, and $\mathcal{A}$ can execute it in full. Whenever a simulated honest member decides $v$ in slot $s$ for the first time, $\mathcal{A}$ queries \textsc{Sign}$(U^\star,k(s),m_s(v))$ and computes all honest shares from the answer by \cref{lem:sharesim}. By capacity, these queries respect the bound of $\psi$ messages per index. $\mathcal{A}$ forwards verification and resolution queries unchanged. The view of $\mathcal{B}$ is distributed exactly as in the real protocol. A certificate on $(s,v)$ that no honest member decided is non-authentic unless $m_s(v)$ collides with a signed message, which happens with probability $\mathsf{Adv}^{\mathrm{CR}}_H$. Its acceptance is therefore the event $\Bad$, which \cref{thm:main} bounds by $Q/q+1/(q-1)$.

\emph{(ii)} follows from (i) and agreement: two accepted certificates on different values for one slot would require an honest decision on each of them.

\emph{(iii).} Honest shares are evaluations of the degree-$t$ vector polynomial $y\mapsto A(y)\bfe(m)$, and at most $t$ shares are wrong or missing. With $r$ received shares, $e$ of them wrong, and $s'$ missing, decoding succeeds when $r\ge t+1+2e$. Since $r=n-s'$ and $e+s'\le t$, this holds for $n\ge3t+1$; the exact decoder applies it to each coordinate. In the fast decoder, a wrong share with error $\mathbf{d}_j\neq\mathbf{0}$ escapes detection only if $\boldsymbol{\rho}\T\mathbf{d}_j=0$, which has probability $1/q$ because $\boldsymbol{\rho}$ is drawn after the set of shares is frozen. A share that arrives later is not part of this decoding; a retry draws a new $\boldsymbol{\rho}$ and is accounted separately. Unless one of the at most $t$ wrong shares escapes, every non-suspect share is correct, and interpolation from $t+1$ of them returns $\balpha$. The stopping rule of online error correction is the standard one~\cite{BCG1993}: an output that agrees with $2t+1$ received shares agrees with at least $t+1$ honest shares and is therefore the honest polynomial.
\end{proof}

\begin{remark}[System-level error budget]\label{rem:budget}
The components fail on separate events, which a deployment must add up. Over the lifetime of a system with $E$ anchor matrices (including those of arbiter-verified keys),
\begin{align*}
\varepsilon_{\mathsf{sys}}\le{}&\underbrace{\tfrac{Q_{\mathrm{total}}}{q}+\tfrac{E}{q-1}}_{\text{USS, incl.\ batches}}+\varepsilon_{\mathsf{MAC}}+\mathsf{Adv}^{\mathrm{CR}}_{H}\\
&+\mathsf{Adv}^{\mathrm{CR}}_{H_{\mathsf{BFT}}}+\varepsilon_{\mathsf{BFT}}+\tfrac{R\,t}{q}+\tfrac{tL_{\mathsf{RS}}}{q},
\end{align*}
where $Q_{\mathrm{total}}$ counts all verification, resolution and optimistic-combiner checks and charges each tested batch by its size (\cref{lem:batch}), $\mathsf{Adv}^{\mathrm{CR}}_{H}$ covers the hashes of signed and certified messages (zero with the injective encoding), $\mathsf{Adv}^{\mathrm{CR}}_{H_{\mathsf{BFT}}}$ the digests used inside agreement, $Rt/q$ the fast-decoder calls, and $tL_{\mathsf{RS}}/q$ the redistributions (\cref{lem:reshare}), with $L_{\mathsf{RS}}$ summed over all of them; $\varepsilon_{\mathsf{MAC}}$ is likewise summed over pairs and epochs. With ideal private channels, the issuance protocols are perfectly secure and add no term; if those channels are built from pads and the same MAC, their checks are counted in $\varepsilon_{\mathsf{MAC}}$. Each term is counted once; in particular, $\varepsilon_{\mathsf{agr}}$ in \cref{thm:tn} stands for the agreement terms of this sum, not for additional events. A fully information-theoretic deployment needs the injective encoding for certificates, the polynomial MAC above, agreement on full values, and private channels from pads (\cref{sec:discussion}).
\end{remark}

\begin{remark}[What TN changes and what it does not]
TN removes signer-side state: indices follow the slot number fixed by consensus, and an honest committee cannot equivocate. The signing key is never in one place, and certificates have constant size. What TN does \emph{not} provide is per-validator attribution: a certificate shows that the committee decided, not who voted. \Cref{sec:av} restores attribution at extra cost.
\end{remark}

\begin{remark}[Linear communication for $t<n/5$]\label{rem:linear}
The following observation is specific to \scheme{}. If the committee accepts $n\ge5t+1$, the certificate can itself be the quorum certificate of a HotStuff-style protocol, without all-to-all authenticated voting. Members share $A_k$ with degree $d=n-2t-1$ and release a share as their vote; the vote message includes the view and the phase. A certificate needs $d+1=n-2t$ shares, and three counts make this work. All $n-t$ honest members suffice for a certificate (liveness). Each certificate contains at least $n-3t$ honest votes, and two certificates for one view and phase would need $2(n-3t)>n-t$ distinct honest voters, i.e., an honest member that voted twice in one view and phase (uniqueness). Decoding a degree-$d$ codeword from $n$ shares with $t$ errors needs $n\ge d+1+2t$ (robustness); a leader that waits only for $n-t$ shares corrects fewer errors and must wait for more on failure. Two conditions differ from TN. First, capacity must be budgeted per (view, phase) message, because honest members may vote in many views. Second, in the soundness argument, shift~(B) becomes $\bg\bh\T L(y)$, where $L$ has degree at most $d$, $L(0)=1$, and vanishes at the at most $d$ points whose shares on the target message the adversary has seen. The vector $\bh$ must be orthogonal to $\bfe(m')$ for \emph{every} message $m'$ on which some honest share was released; therefore all such messages, not only certified ones, count towards $\psi$. We leave a full treatment of this variant to future work.
\end{remark}

\textbf{Comparison with threshold signatures.} TN is structurally close to a threshold signature: the committee holds a degree-$t$ sharing of a linear signing map, and any $t+1$ correct shares interpolate a certificate. \Cref{tab:threshold} compares it with BLS~\cite{CKS2005,SBFT2019} and lattice~\cite{ThRaccoon2024,Ringtail2025} threshold signatures. TN rests on no hardness assumption, and certificates issued today cannot be forged by future cryptanalysis. The price is verification with secret keys of admitted verifiers, which is unavoidable because no publicly verifiable signature is information-theoretically unforgeable; transferability is preserved because validity is a global predicate (\cref{cor:transfer}). A TN share is one matrix--vector product released after the consensus decision, with no signing rounds among the signers and no per-share verification: wrong shares are removed by Reed--Solomon decoding. The new ingredient is that the virtual anchors make a secret-shared linear map transferable and non-repudiable; without them, a secret-shared MAC would let every verifier forge.

\begin{table}[!t]
\centering
\caption{Threshold notarization versus threshold signatures and distributed quorum signatures (DQS).}
\label{tab:threshold}
\scriptsize
\setlength{\tabcolsep}{2.5pt}
\resizebox{\columnwidth}{!}{%
\begin{tabular}{@{}lllll@{}}
\toprule
 & \textbf{BLS thr.}~\cite{CKS2005,SBFT2019} & \textbf{Lattice thr.}~\cite{ThRaccoon2024,Ringtail2025} & \textbf{DQS}~\cite{KSW2026DQS} & \textbf{\scheme{} TN}\\
\midrule
Security basis & pairings (not PQ) & lattices & PQ sig. & none (IT)\\
Verification & public key & public key & public key & secret keys\\
Signing rounds & 1 + checks & 2--3 & approvals & 0 (decided)\\
Faulty shares & verified & by protocol & n/a & RS decoding\\
Size & 96\,B & several KB & $O(n)$ sigs & $32(\omega{+}1)$\,B\\
Key lifetime & unlimited & unlimited & per scheme & $\psi$ msgs / index\\
Individual signers & no & no & some node & yes (IV, AV)\\
\bottomrule
\end{tabular}}
\end{table}

\subsection{Arbiter-Verified Keys and Batch Arbitration}\label{sec:av}
Setting $\omega=0$ gives a key whose only verifier is the arbiter. The anchor matrix becomes a scalar $\gamma$, a signer-index is a row vector $\mathbf{a}_{U,k}\in\F^{1\times(\psi+1)}$, and a signature is the single element $\alpha=\mathbf{a}_{U,k}\bfe(m)$. The committee keeps $\sh{\gamma}$ and $\sh{\mathbf{b}_{0,U,k}}=\sh{\gamma}\sh{\mathbf{a}_{U,k}}$ and opens $\mathbf{a}_{U,k}$ to the signer. No verifier key is issued. \Cref{thm:main} holds verbatim with $C=\emptyset$: shift (A) uses $\bu=1$, and shift (B) needs only $\gamma\neq0$.

Checking each signature with $\Res$ costs one multiplication and one opening online, plus a preprocessed non-zero mask (\cref{sec:construction}). A block of signatures can instead be checked at once. Note that $\gamma$ is uniform on $\F$ and is therefore zero with probability $1/q$; this is the singular-anchor term of \cref{thm:main} for $\omega=0$, and it is shared by all checks on the same $\gamma$.

\begin{lemma}[Batch arbitration]\label{lem:batch}
Let $\sigma_1,\dots,\sigma_L$ be arbiter-verified signatures, fixed before a uniform $r\in\F$ is opened. Let $\sh{\delta_j}=\sh{\gamma}\alpha_j-\sh{\mathbf{b}_{0,U_j,k_j}}\bfe(m_j)$ (local), $\sh{D}=\sum_{j=1}^{L} r^{j-1}\sh{\delta_j}$, and output $[\rho D=0]$ after one multiplication by a fresh uniform non-zero $\sh{\rho}$ and one opening. The output reveals only the verdict. If some $\sigma_j$ is not acceptable, then on any transcript of the ideal game of the proof of \cref{thm:main} (in particular, for a first query), the verdict is $1$ with probability at most $(L+1)/q$: at most $1/q$ for the event $\gamma=0$, and at most $L/q$ otherwise.
\end{lemma}
\begin{proof}
$\rho D$ is $0$ iff $D=0$, and otherwise it is uniform on $\F^*$; the verdict therefore determines the distribution of the opened value, which can be simulated from the verdict. Suppose $\sigma_{j_0}$ is not acceptable. The event $\gamma=0$ has probability $1/q$; assume $\gamma\neq0$. \emph{Case 1: some $\delta_j\neq0$.} Then $D(r)=\sum_j r^{j-1}\delta_j$ is a non-zero polynomial in $r$ of degree at most $L-1$. The batch and all keys are fixed before $r$ is opened; hence $r$ is uniform and independent of the $\delta_j$, and $D(r)=0$ with probability at most $(L-1)/q$. \emph{Case 2: all $\delta_j=0$.} In particular, the arbiter's check accepts the non-acceptable $\sigma_{j_0}$. With $\omega=0$ there are no corrupted functionals. If $\sigma_{j_0}$ deviates by a known non-zero amount $\Delta$ from a valid or returned signature (a corrupted signer, or a modified signature of an honest one), then $\delta_{j_0}=\gamma\Delta\neq0$, which is impossible. If $\sigma_{j_0}$ is a fresh message of an honest signer, then $\delta_{j_0}=\gamma(\alpha_{j_0}-\mathbf{a}\bfe(m_{j_0}))$ vanishes iff $\mathbf{a}\bfe(m_{j_0})=\alpha_{j_0}$, which has probability $1/q$ by shift (B), because on such a transcript $\mathbf{a}$ is uniform on its affine posterior. Altogether, $1/q+(L-1)/q+1/q=(L+1)/q$.
\end{proof}
Weights starting at $r^{1}$, i.e., $r^{j}$ for $j=1,\dots,L$, do not give this bound: for $L=1$ and a fresh unsigned message submitted with $\alpha=0$, the verdict is $[r\gamma\mathbf{a}\bfe(m)=0]$, which holds with probability $3/q-3/q^2+1/q^3>2/q$. An exhaustive check over $\mathbb{F}_3,\mathbb{F}_5,\mathbb{F}_7$ (\texttt{batch\_check.py} in~\cite{SILMVcode}) confirms the counterexample, and the corrected bound for batches of $L\in\{1,2\}$ unsigned honest items and every choice of submitted values.

\textbf{Batches within a run.} As in the proof of \cref{thm:main}, we answer every batch in the ideal game with the bit ``every item is acceptable'', which the adversary can compute itself. Each batch draws its own $r$ after it is fixed and its own fresh $\rho$; $r$ is public and independent of all secrets; hence it does not change the posterior. With $Q_{\mathrm{ind}}$ individual resolution queries and $B$ batches of sizes $L_1,\dots,L_B$ on the same $\gamma$, the proof above gives
\[
\Pr[\Bad]\le\frac{1+Q_{\mathrm{ind}}+\sum_{b=1}^{B}L_b}{q}:
\]
the event $\gamma=0$ is counted once, each individual query contributes $1/q$, and each non-acceptable batch contributes $(L_b-1)/q$ for cancellation and $1/q$ for one non-acceptable item. A rejected batch is bisected to locate the offenders. Bisection identifies the items that fail the algebraic check; it identifies \emph{every} non-acceptable item only outside the event $\Bad$, since a sub-batch can pass by cancellation and an unsigned item can pass by a lucky guess. Locating $\beta$ offenders takes $O(\beta\log L)$ further tests, but the error accounting charges the sizes of all tested sub-batches: the bound above holds with $\sum_bL_b$ taken over every tested batch, which can reach $\Theta(L\log L)$ item appearances when most items are bad. Each test uses fresh $r$ and $\rho$.

\textbf{Clients.} A client receives arbiter-verified indices. Its transaction carries a $32$-byte signature, validators verify the transactions of a block collectively with \cref{lem:batch}, and the block's TN certificate attests to the outcome for external parties. Non-repudiation holds through $\Res$. One signer-index costs $(\psi+1)(r_c+\mu_c+n)$ field elements in total, i.e., a constant number per validator and transaction. The client receives $(\psi+1)n$ elements per index, about $32n$ bytes per transaction, which is the only client-side cost that grows with $n$ (\cref{sec:eval:scale}).

\textbf{Accountable TN.} Validators may additionally attach an arbiter-verified signature to their commit votes. The committee batch-checks the votes of each block. Two accepted signatures of one validator on conflicting votes for the same slot are equivocation evidence in the sense of \cref{thm:account}, whose proof applies unchanged with $\omega=0$. This restores per-validator accountability on top of TN, at the cost quantified in \cref{sec:eval:scale}.

\subsection{Open Verifier Sets}\label{sec:enrol}
\begin{proposition}[Dynamic and retroactive enrollment]\label{prop:enrol}
A verifier with a fresh identifier $x_{\mathrm{new}}\notin\{0,x_1,\dots,x_\nu\}$ can be enrolled at any time by opening $\bc_{\mathrm{new}}=\Gamma\T\blam(x_{\mathrm{new}})$ and $\bb_{\mathrm{new},k}=M_k\T\blam(x_{\mathrm{new}})$ for any set of indices $k$, including indices whose slots are long past. This uses only local linear maps and robust openings, costs $(\omega+1)n$ field elements once per epoch for $\bc_{\mathrm{new}}$ and $(\psi+1)n$ per index, and needs no multiplication. Consider the following extension of \cref{def:exp}: the committee corruption remains static, but the adversary may adaptively request enrollments of fresh identifiers of its choice and adaptively corrupt enrolled verifiers, receiving all keys the corrupted verifier holds, for past and future indices. \Cref{thm:main,thm:tn} continue to hold as long as the \emph{cumulative} number of verifiers ever corrupted for the same anchor matrix $\Gamma$, including verifiers whose credentials were later revoked, is at most $\omega$.
\end{proposition}
\begin{proof}
The keys are exactly those $\Fki$ would have issued to $x_{\mathrm{new}}$. The bound of \cref{thm:main} does not depend on $\nu$, only on $|C|\le\omega$ and on the distinctness of identifiers used in \eqref{eq:indep}. For the adaptive extension, we follow the proof of \cref{thm:main} with $C_j$, the set of verifiers corrupted before query $j$. Step~1 is unchanged, because the ideal answers do not depend on keys. For Step~2, we argue by induction over the transcript: a corruption of $V_s$ reveals $\bc_s\T=\blam(x_s)\T\Gamma$ and $\bb_{s,U,k}\T=\bc_s\T A_{U,k}$ for all $(U,k)$, i.e., it appends one row to $K$ and the corresponding row to every $B_{U,k}$. These are linear constraints of the same form; hence the posterior remains uniform on a product of affine subspaces. Step~3 applies with $\Lambda_{C_j}$, which has at most $\omega$ columns; \eqref{eq:indep} holds for every honest $x_i$ because identifiers are distinct, whatever order they were chosen in. Revoking a credential does not remove its row from $K$, which is why the bound is cumulative.
\end{proof}
Enrollment must be admission-controlled, e.g., by the committee's membership rules. Since $\omega+1$ colluding verifiers forge (\cref{prop:tight}(b)), Sybil verifiers must not be admitted. With a fresh $\Gamma$ per epoch, $\omega$ bounds the verifiers corrupted \emph{per epoch}, cumulatively over the epoch and over any period in which that epoch's state is retained. This adaptive statement concerns verifiers only; the committee corruption stays static, as in \cref{thm:dki}.

\subsection{Evolving Committees and Proactive Refresh}\label{sec:reshare}
Let the shared state, i.e., $\sh{\Gamma}$, the live $\sh{A_k}$ and the retained $\sh{M_k}$, be held by a committee $\mathcal{P}$ ($n\ge3t+1$) and be handed over to a committee $\mathcal{P}'$ ($n'\ge3t'+1$), which may overlap with $\mathcal{P}$ or differ from it in size.

\textbf{Protocol $\Pi_{\mathsf{RS}}$.}
(i)~Every old member $P_j$ shares each of its shares with $\mathcal{P}'$ by a perfect VSS of degree $t'$~\cite{BGW1988,BTH2008}. The VSS guarantees that each dealing is a consistent degree-$t'$ sharing, but a corrupted dealer may share a wrong value. (ii)~$\mathcal{P}'$ opens a uniform $r$ after all dealings are complete, and forms, for each dealer $j$, the sharing $\sh{w_j}=\sum_{\ell=1}^{L}r^{\ell}\sh{s^{(\ell)}_j}$ of the combination of the values $s^{(\ell)}_j$ it dealt for the $L$ secrets. (iii)~With a parity-check matrix $H$ of the degree-$t$ Reed--Solomon code on the old points, $\mathcal{P}'$ computes $\sh{H\mathbf{w}}$ locally and opens it. (iv)~The syndrome $H\mathbf{w}$ is decoded to the support $\mathcal{D}$ of the error vector, i.e., the dealers that deviated. Silent dealers are erasures. (v)~The new share of every secret is $\sum_{j\in G}\lambda^G_j\sh{s^{(\ell)}_j}$ for a fixed set $G$ of $t+1$ dealers outside $\mathcal{D}$, with Lagrange coefficients $\lambda^G_j$ at $0$. Old members erase their shares.

\begin{lemma}[Robust redistribution]\label{lem:reshare}
Against a static adversary that corrupts at most $t$ members of $\mathcal{P}$ and at most $t'$ members of $\mathcal{P}'$, $\Pi_{\mathsf{RS}}$ outputs degree-$t'$ sharings of the \emph{same} secrets except with probability at most $tL/q$, where $L$ is the number of secrets. Privacy is perfect, also on that event: the adversary's complete view of one transition, i.e., the old shares of corrupted old members, all messages that corrupted parties receive in the dealings, the public $r$, the opened syndrome polynomials and the new shares of corrupted new members, has the same distribution for all values of the secrets, and a simulator that knows none of the secrets generates it exactly.
\end{lemma}
\begin{proof}
Write $\mathbf{w}=\mathbf{c}+\boldsymbol{\epsilon}$, where $\mathbf{c}$ is the combination of the correct shares, a codeword, and $\epsilon_j=\sum_\ell r^\ell\Delta^{(\ell)}_j$ collects the deviations $\Delta^{(\ell)}_j$ of dealer $j$. Then $H\mathbf{w}=H\boldsymbol{\epsilon}$ depends only on deviations that the adversary chose; opening it reveals nothing about the secrets. Since at most $t$ dealers deviate and the code has minimum distance $n-t\ge2t+1$, syndrome decoding recovers $\boldsymbol{\epsilon}$, and with it every dealer with $\epsilon_j\neq0$. A deviating dealer escapes only if $\epsilon_j=0$ although some $\Delta^{(\ell)}_j\neq0$; this is a non-zero polynomial of degree at most $L$ in $r$, which has probability at most $L/q$, and at most $t$ dealers deviate. All dealers in $G$ then shared their correct shares, and the Lagrange combination of $t+1$ correct shares is the secret.

\emph{Privacy.} Giving the adversary more corrupted new members only enlarges its view; we therefore assume $|T'|=t'$. An honest dealer $h$ shares $s^{(\ell)}_h$ with $P^{(\ell)}_h(z)=s^{(\ell)}_h+\sum_{\kappa=1}^{t'}a^{(\ell)}_{h,\kappa}z^{\kappa}$, whose coefficients $\mathbf{a}^{(\ell)}_h$ are fresh and uniform. Let $V=[\eta_i'^{\,\kappa}]_{i\in T',\kappa\in[t']}$, which is invertible because the new points are distinct and non-zero. The values that corrupted new members receive are $\mathbf{e}^{(\ell)}_h=s^{(\ell)}_h\mathbf{1}+V\mathbf{a}^{(\ell)}_h$, which are uniform on $\F^{t'}$ and independent of the secrets, of the old shares and of everything else, because $\mathbf{a}^{(\ell)}_h$ is fresh. By the privacy of perfect VSS, the other dealing messages that corrupted parties receive are simulatable from these values~\cite{BGW1988}. Consider the opened syndrome polynomial $S_i(z)=\sum_j H_{ij}\sum_\ell r^\ell P^{(\ell)}_j(z)$. Its constant term is $(H\boldsymbol{\epsilon})_i$, because the correct old shares form a codeword, and the adversary knows it. Its other coefficients have the honest part $\sum_h H_{ih}\sum_\ell r^\ell\mathbf{a}^{(\ell)}_h=V^{-1}\big(\sum_h H_{ih}\sum_\ell r^\ell\mathbf{e}^{(\ell)}_h-\mathbf{1}\sum_h H_{ih}\sum_\ell r^\ell s^{(\ell)}_h\big)$, and $\sum_h H_{ih}\sum_\ell r^\ell s^{(\ell)}_h=-\sum_{j\in T}H_{ij}\sum_\ell r^\ell s^{(\ell)}_j$, again because correct old shares form a codeword. The adversary knows the right-hand side, since it holds the corrupted old shares. Hence the syndrome polynomials are a fixed function of the $\mathbf{e}^{(\ell)}_h$, of $r$ and of the adversary's own data. The same holds for the set $G$, which syndrome decoding derives from the constant terms, and for the new shares $\sum_{j\in G}\lambda^G_jP^{(\ell)}_j(\eta'_i)$ of corrupted new members. The corrupted old shares are uniform and independent of the secrets, because there are at most $t$ of them. The complete view is therefore a fixed function of independent variables whose distribution does not depend on the secrets, and the simulator samples them and runs the adversary.
\end{proof}
The lemma does not claim that the adversary's old and new views are independent given the secrets. A corrupted dealer may reshare its own old share $Z$ with the constant polynomial $Z$; both views then contain $Z$. That correlation consists only of information the adversary already holds, and the joint view is still simulatable.

\textbf{Consequences.} (a) \emph{Membership changes} at epoch boundaries need no resharing of signing state: the new committee runs issuance for the next epoch with a fresh $\Gamma$, and only the retained state of (c) is reshared. Mid-epoch changes use $\Pi_{\mathsf{RS}}$. (b) \emph{Proactive refresh}~\cite{Herzberg1995}: refreshing with $\mathcal{P}'=\mathcal{P}$ at every epoch, and erasing old shares, is \emph{designed} to ensure that forging requires $t+1$ members corrupted \emph{within one epoch}, rather than $t+1$ members corrupted at any time. \Cref{lem:reshare} covers a single transition against a static adversary; it does not by itself give security against a mobile adversary, and a formal treatment of that model is future work. Two limits apply even then. Refreshing the sharings of a retained $\Gamma$ does not change the verifier functionals $\bc_i$ already issued; hence $\omega$ bounds the verifiers corrupted cumulatively for that $\Gamma$ (\cref{prop:enrol}), not those corrupted at one time. Erasure by honest members does not remove copies that a corrupted member kept. (c) \emph{Retained state}: verifying history needs no retained secrets when later certificates cover hashes of earlier blocks. This relies on collision resistance of the block hash and on an authenticated current checkpoint, and is therefore computational; it does not replace the information-theoretic verification of the old certificates. Only the dispute window and the indices that the committee chooses to keep open for retroactive enrollment are reshared.

\subsection{Cost of Threshold Notarization}\label{sec:tn:cost}
With $E=(\omega+1)(\psi+1)$ and $r_c,\mu_c$ as in \eqref{eq:dkicost}, one committee index costs
\begin{equation}\label{eq:tncost}
C^{\mathsf{TN}}_{\mathrm{idx}}=E\,r_c+E\,\mu_c+\nu(\psi+1)\,n
\end{equation}
field elements, and each certified slot adds $n(\omega+1)$ elements of shares per combiner. Per member and slot, this is
\[
(\omega+1)\Big[\big(1+\tfrac{1}{\psi}\big)\tfrac{r_c+\mu_c}{n}+1\Big]+\nu\big(1+\tfrac1\psi\big)=O(\omega+\nu),
\]
because $r_c/n$ and $\mu_c/n$ are constants for $d=\Theta(n)$. It is bounded independently of $n$. The only per-member term that grows with $n$ is the authentication of consensus messages themselves, $O(n)$ MAC tags per slot, which any all-to-all BFT protocol pays. Our simulator, which runs $\Pi_{\mathsf{TN}}$ with garbage, silent and equivocating members, matches \eqref{eq:tncost}, instantiated with the semi-honest extraction rate $d=n-t$, to within $0.8\%$ (\cref{sec:eval:scale}).

\section{Implementation and Evaluation}\label{sec:eval}

\subsection{Setup}
We implemented the \scheme{} core in C: field arithmetic over $\mathbb{F}_{2^{255}-19}$ in radix $2^{51}$, signing with a log-depth power ladder for $\bfe(m)$ and lazily reduced inner products, and verification. Long messages are mapped to $\F$ with OpenSSL SHA3-256. We also wrote three Python components:
a reference implementation for arbitrary primes $q$; a simulator of $\Pki$ that counts every field element sent and decodes outputs with Berlekamp--Welch against $t$ parties sending corrupted shares; a Monte-Carlo attack suite.
All timings in this section were measured on the \emph{same host}, an Intel Core i9-13900HX under WSL2 (Ubuntu, GCC~15.2 \texttt{-O3 -march=native}, Rust~1.99), single-threaded. Baselines use the fastest available implementations: Ed25519 from \texttt{ed25519-dalek}, and ML-DSA-44/65, Falcon-512 and SLH-DSA from the AVX2-optimized PQClean code through the Rust \texttt{pqcrypto} crates. For SLH-DSA we measure SPHINCS$^+$-SHA2-128f/s-simple, which uses the same parameter sets. All code, simulators and raw data are publicly available at \url{https://github.com/hassankhodaiemehr/silmarils-mv}~\cite{SILMVcode} (\cref{app:artifact}).

\subsection{Microbenchmarks}
\begin{table}[!t]
\centering
\caption{\scheme{} core over $\mathbb{F}_{2^{255}-19}$ ($\mu$s, excluding hashing; SHA3-256 of a $128$\,B message adds $0.52\,\mu$s)}
\label{tab:micro}
\footnotesize
\setlength{\tabcolsep}{4pt}
\begin{tabular}{@{}rr|rr|rr|rr@{}}
\toprule
 & & \multicolumn{2}{c|}{$\psi=1$} & \multicolumn{2}{c|}{$\psi=16$} & \multicolumn{2}{c}{$\psi=64$}\\
$\omega$ & $|\sigma|$ (B) & sign & verify & sign & verify & sign & verify\\
\midrule
1 & 64 & 0.06 & 0.08 & 0.59 & 0.45 & 2.41 & 2.13\\
3 & 128 & 0.12 & 0.10 & 1.02 & 0.47 & 4.00 & 1.70\\
10 & 352 & 0.29 & 0.18 & 2.48 & 0.55 & 8.91 & 1.57\\
20 & 672 & 0.60 & 0.32 & 4.57 & 0.63 & 17.43 & 1.81\\
33 & 1088 & 0.87 & 0.47 & 7.81 & 0.83 & 27.43 & 1.82\\
66 & 2144 & 1.74 & 0.77 & 13.67 & 1.13 & 52.47 & 2.30\\
\bottomrule
\end{tabular}
\end{table}

\Cref{tab:micro} reports signing and verification times. Verification costs $\omega+\psi+2$ multiplications (about $12.5$\,ns each on this host). For $\psi\le 16$ it stays near or below one microsecond up to $\omega=33$, i.e., committees of $100$ validators. Signing costs $(\omega+1)(\psi+1)$ multiplications. With $\psi=1$ it is a few hundred nanoseconds; larger $\psi$ trades signing time for fewer indices and smaller verifier storage per signature.

\begin{table}[!t]
\centering
\caption{Same-host comparison. For \scheme{}: $\omega=10$ ($n=31$), $\psi=16$, times including SHA3-256. ``Verifier key'' is per signer-index ($\psi$ signatures).}
\label{tab:baselines}
\footnotesize
\setlength{\tabcolsep}{5pt}
\resizebox{\columnwidth}{!}{%
\begin{tabular}{@{}lrrrrl@{}}
\toprule
Scheme$^{\ast}$ & $|\sigma|$ (B) & Ver.\ key (B) & Sign ($\mu$s) & Verify ($\mu$s) & Basis\\
\midrule
\scheme{} & 352 & 544$^{\dagger}$ & 3.0 & 1.1 & IT\\
Ed25519 & 64 & 32 & 13.8 & 31.3 & ECDLP\\
ML-DSA-44 & 2420 & 1312 & 64.9 & 23.2 & MLWE/MSIS\\
ML-DSA-65 & 3309 & 1952 & 104 & 37.4 & MLWE/MSIS\\
Falcon-512 & 666 & 897 & 160 & 29.2 & NTRU-SIS\\
SLH-DSA-128f & 17088 & 32 & 6609 & 568 & hash\\
SLH-DSA-128s & 7856 & 32 & 146158 & 224 & hash\\
\bottomrule
\multicolumn{6}{@{}l@{}}{\scriptsize $^{\dagger}$Secret, per signer-index; $34$\,B per signature, plus $352$\,B $\bc_i$ per epoch.}
\end{tabular}}
\par\vspace{2pt}\parbox{\columnwidth}{\scriptsize $^{\ast}$The schemes compared in this table operate under different verification models. \scheme{} requires secret verifier credentials and authorization-controlled verification, whereas standardized signature schemes provide public verification.}
\end{table}

\Cref{tab:baselines} compares \scheme{} with standardized schemes on the same host. Against AVX2-optimized implementations, verification is $22\times$ faster than ML-DSA-44 and $27\times$ faster than Falcon-512, and signing is $22\times$ and $53\times$ faster, respectively. The price is secret, per-signature verifier key material: $34$\,B per future signature, issued offline. The compared schemes address different deployment models. Standardized digital signatures provide unrestricted public verification, whereas \scheme{} assumes admission-controlled verification using verifier credentials. Accordingly, the comparison should be interpreted as a resource comparison rather than a claim of functional equivalence.

\subsection{Ledger-Level Evaluation}
We model the leader-based protocol of \cref{sec:ledger} with $n=3t+1$, $\omega=t$ and $\psi=16$. Per block, $n$ votes go to the leader, which broadcasts a QC of $2t+1$ votes plus its proposal signature to $n-1$ replicas. Each replica verifies the QC and the proposal. We count authentication bytes only and use the measured times of \cref{tab:micro,tab:baselines}. For \scheme{} we report both individual-vote QCs and aggregated QCs (\cref{cor:agg}). Offline issuance cost follows \eqref{eq:dkicost}, with $n+1$ signatures per block (happy path: one vote per validator and one proposal; protocols with several voting phases scale this term proportionally).

\begin{figure*}[!t]
\centering
\begin{tikzpicture}
\begin{groupplot}[group style={group size=4 by 1, horizontal sep=1.25cm}, width=0.262\textwidth, height=4.6cm,
  xmin=3, xmax=102, xlabel={committee size $n$}, ymode=log, title style={font=\footnotesize, yshift=-1ex},
  every axis plot/.append style={mark size=1.4pt}, legend style={font=\scriptsize}]
\nextgroupplot[title={(a) QC size (KiB)}, ymin=0.03, ymax=1000]
\addplot[cSLH, mark=triangle*] table[x=n,y=slhqc]{data_ledger.dat};
\addplot[cD44, mark=square*] table[x=n,y=d44qc]{data_ledger.dat};
\addplot[cFAL, mark=diamond*] table[x=n,y=falqc]{data_ledger.dat};
\addplot[cED, mark=x, dashed] table[x=n,y=edqc]{data_ledger.dat};
\addplot[cSIL, mark=*] table[x=n,y=silqc]{data_ledger.dat};
\addplot[cAGG, mark=pentagon*, very thick] table[x=n,y=silaggqc]{data_ledger.dat};
\nextgroupplot[title={(b) leader egress / block (KiB)}, ymin=0.2, ymax=100000]
\addplot[cSLH, mark=triangle*] table[x=n,y=slhlead]{data_ledger.dat};
\addplot[cD44, mark=square*] table[x=n,y=d44lead]{data_ledger.dat};
\addplot[cFAL, mark=diamond*] table[x=n,y=fallead]{data_ledger.dat};
\addplot[cED, mark=x, dashed] table[x=n,y=edlead]{data_ledger.dat};
\addplot[cSIL, mark=*] table[x=n,y=sillead]{data_ledger.dat};
\addplot[cAGG, mark=pentagon*, very thick] table[x=n,y=silagglead]{data_ledger.dat};
\nextgroupplot[title={(c) QC + proposal verification ($\mu$s)}, ymin=1, ymax=100000]
\addplot[cSLH, mark=triangle*] table[x=n,y=slhcpu]{data_ledger.dat};
\addplot[cD44, mark=square*] table[x=n,y=d44cpu]{data_ledger.dat};
\addplot[cFAL, mark=diamond*] table[x=n,y=falcpu]{data_ledger.dat};
\addplot[cED, mark=x, dashed] table[x=n,y=edcpu]{data_ledger.dat};
\addplot[cSIL, mark=*] table[x=n,y=silcpu]{data_ledger.dat};
\addplot[cAGG, mark=pentagon*, very thick] table[x=n,y=silaggcpu]{data_ledger.dat};
\nextgroupplot[title={(d) total per validator / block (KiB)}, ymin=0.1, ymax=3000,
  legend to name=ledgerlegend, legend columns=3, legend style={/tikz/every even column/.append style={column sep=6pt}}, legend cell align=left]
\addlegendimage{cSLH, mark=triangle*}\addlegendentry{SLH-DSA-128s}
\addplot[cD44, mark=square*] table[x=n,y=d44]{data_party.dat}; \addlegendentry{ML-DSA-44}
\addplot[cFAL, mark=diamond*] table[x=n,y=fal]{data_party.dat}; \addlegendentry{Falcon-512}
\addplot[cED, mark=x, dashed] table[x=n,y=ed]{data_party.dat}; \addlegendentry{Ed25519 (not PQ)}
\addlegendimage{cSIL, mark=*}\addlegendentry{\scheme{} indiv. (a--c)}
\addlegendimage{cAGG, mark=pentagon*, very thick}\addlegendentry{\scheme{} agg. (a--c)}
\addplot[cSIL, mark=*] table[x=n,y=silaggdki]{data_party.dat}; \addlegendentry{agg.\ + $\Pki$, malicious (d)}
\addplot[cSIL, mark=o, dotted] table[x=n,y=silaggdkish]{data_party.dat}; \addlegendentry{agg.\ + $\Pki$, semi-honest (d)}
\addplot[cAGG, mark=pentagon*, very thick] table[x=n,y=silaggtd]{data_party.dat}; \addlegendentry{agg.\ + dealer (d)}
\end{groupplot}
\end{tikzpicture}

\vspace{1mm}\pgfplotslegendfromname{ledgerlegend}
\caption{Ledger-level model of a leader-based BFT committee ($n=3t+1$, $\omega=t$, $\psi=16$): (a) quorum-certificate size; (b) online authentication bytes sent by the leader per block; (c) per-replica time to verify one QC and one proposal signature; (d) average total bytes per validator per block, \emph{including} offline key issuance by the committee ($\Pki$, both cost models of \eqref{eq:dkicost}) or by a dealer/HSM. Aggregated \scheme{} QCs contain $\omega+1$ field elements, an $n$-bit bitmap and a $4$-byte index per signer; individual votes also carry their $4$-byte index.}
\label{fig:ledger}
\end{figure*}
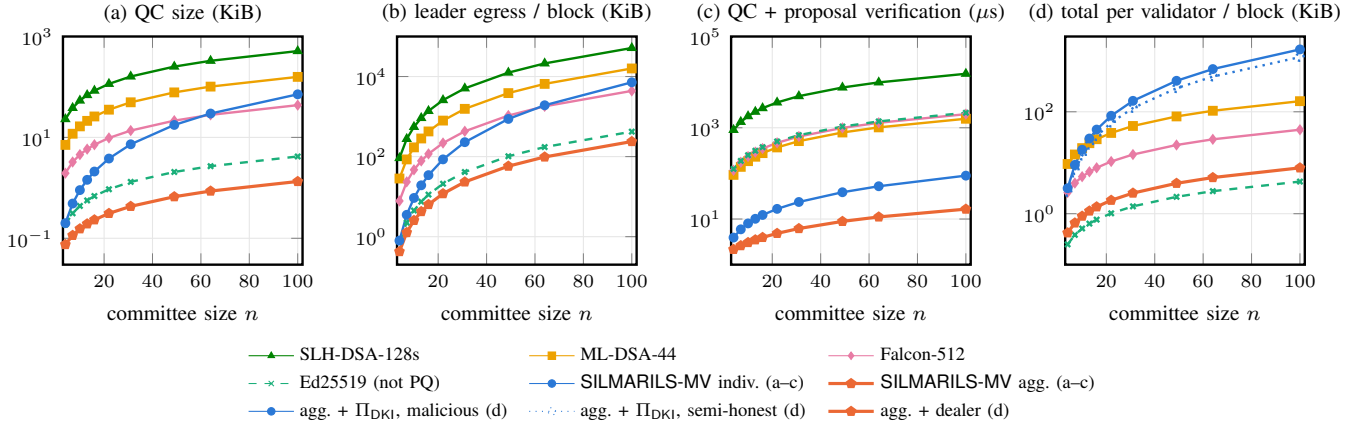

\textbf{QC size.} \Cref{fig:ledger}(a) shows QC size against committee size. Individual-vote QCs grow as $(2t+1)(t+1)$ field elements; they beat ML-DSA-44 by $7\times$ at $n=31$ and match Falcon-512 at $n\approx 60$. Aggregated QCs grow only linearly. At $n=31$ an aggregated QC is $440$\,B ($352$\,B for the vector, $4$\,B of bitmap and $84$\,B for the $21$ signer indices), against $49.6$\,KiB for ML-DSA-44, $13.7$\,KiB for Falcon-512 and $1.3$\,KiB for Ed25519 without aggregation.

\textbf{Leader bottleneck.} \Cref{fig:ledger}(b) plots online authentication bytes sent by the leader per block. Note that this compares \emph{aggregated} \scheme{} QCs with the \emph{non-aggregated} QCs that standardized post-quantum schemes support natively; SNARG-based aggregation of Falcon~\cite{Zhaolu2025} shrinks QCs at a large proving cost and is not modelled. With aggregation, the leader sends $23.3$\,KiB per block at $n=31$, against $1.52$\,MiB for ML-DSA-44 ($67\times$ less) and $429$\,KiB for Falcon-512 ($18\times$ less). It even undercuts non-aggregated Ed25519 at every committee size we consider.

\textbf{Verification time.} \Cref{fig:ledger}(c) shows the per-replica cost of verifying one QC and one proposal signature. At $n=31$ it is $6.2\,\mu$s for aggregated \scheme{}, against $0.51$\,ms for ML-DSA-44, $0.64$\,ms for Falcon-512 and $0.69$\,ms for Ed25519, an improvement of $82\times$ over ML-DSA-44.

\textbf{Total cost including issuance.} \Cref{fig:ledger}(d) adds the offline key material to the online traffic. With a dealer or HSM (e.g., the membership service of Fabric~\cite{Fabric2018}), the total is $2.5$\,KiB per validator per block at $n=31$, $21\times$ less than ML-DSA-44. With dealer-free issuance by $\Pki$, costed with the malicious estimate $d=n-2t$ of \cref{sec:dki}, the average total per validator is \emph{lower} than ML-DSA-44 for $n\le10$ and higher for larger committees: $164$\,KiB versus $53$\,KiB at $n=31$ ($118$\,KiB under the semi-honest cost model that our simulator validates). The malicious figure omits the fault-detection and broadcast-emulation terms listed in \cref{sec:dki}. This is the price of removing every trusted party and every computational assumption. The extra traffic is offline, evenly spread over validators, and does not touch the leader bottleneck. At one block per second and $n=31$ it amounts to $1.3$\,Mbit/s per validator, and to $573$\,MiB of issuance traffic and $3.7$\,MiB of verifier key storage per validator per one-hour epoch. Any USS consumes key material per signature~\cite{Amiri2018,SwansonStinson2016}, and QDS consume it through quantum channels~\cite{Yin2023NSR,Li2023OTUH}. For \scheme{} the cost of producing that key material without any trusted party or quantum link is given in closed form by \eqref{eq:dkicost}.

\subsection{Scalability to 1000 Validators}\label{sec:eval:scale}
We now evaluate the modes of \cref{sec:scale} for $n=3t+1$ from $4$ to $1000$. TN uses $\omega=16$, $\psi=64$ and $\nu=64$ external verifiers, which receive every certificate, and issuance is costed with the malicious model $d=n-2t$. Consensus in TN mode runs two all-to-all voting rounds authenticated by $32$-byte information-theoretic MAC tags, with one tag per delivered message: each receiver gets only its own entry of the PBFT authenticator. The model is an analytic \emph{normal-case} estimate under declared exclusions: it omits pre-prepare, client, checkpoint and view-change traffic, and its issuance term uses the malicious estimate of \cref{sec:dki}, which omits fault-detection and broadcast-emulation terms. It is not the measured cost of a fully implemented malicious-secure ledger. The baselines are post-quantum signatures in a HotStuff-style star topology (one vote to the leader, a QC of $2t+1$ signatures to everyone) and in a PBFT/Tendermint-style all-to-all topology (two voting rounds), plus Ed25519 and BLS~\cite{Boneh2018BLS} (analytic, $96$-byte aggregate plus bitmap) as non-post-quantum references. As before, we count authentication bytes only.

\begin{figure}[!t]
\centering
\begin{tikzpicture}
\begin{loglogaxis}[width=\columnwidth, height=6.0cm, xlabel={committee size $n$}, ylabel={avg.\ bytes per validator per block (KiB)},
  xmin=3.5, xmax=1150, ymin=0.15, ymax=300000, legend columns=2, legend style={at={(0.5,-0.25)},anchor=north}, legend cell align=left,
  xtick={4,10,31,100,301,1000}, xticklabels={4,10,31,100,301,1000}]
\addplot[cSIL, mark=*] table[x=n,y=ivp]{data_scale.dat}; \addlegendentry{\scheme{} IV, dealer-free}
\addplot[cD44, mark=square*, dashed] table[x=n,y=d44ap]{data_scale.dat}; \addlegendentry{ML-DSA-44, all-to-all}
\addplot[cD44, mark=square*] table[x=n,y=d44sp]{data_scale.dat}; \addlegendentry{ML-DSA-44, star}
\addplot[cFAL, mark=diamond*] table[x=n,y=falsp]{data_scale.dat}; \addlegendentry{Falcon-512, star}
\addplot[cD65, mark=triangle*] table[x=n,y=tnap]{data_scale.dat}; \addlegendentry{\scheme{} TN + accountability}
\addplot[cAGG, mark=pentagon*, very thick] table[x=n,y=tnp]{data_scale.dat}; \addlegendentry{\scheme{} TN}
\addplot[gray, no marks, dotted, thick] table[x=n,y=macp]{data_scale.dat}; \addlegendentry{MAC-only BFT (no certificate)}
\addplot[cED, mark=x, dashed] table[x=n,y=edsp]{data_scale.dat}; \addlegendentry{Ed25519, star (not PQ)}
\addplot[cRED, mark=+, dashed] table[x=n,y=blsp]{data_scale.dat}; \addlegendentry{BLS aggregate (not PQ)}
\end{loglogaxis}
\end{tikzpicture}
\caption{Average authentication bytes per validator per block, including all offline key issuance, up to $n=1000$. One-time MAC key material for TN consensus is not counted as traffic (see text). TN costs an amount bounded independently of $n$ on top of MAC-authenticated consensus.}
\label{fig:scale}
\end{figure}
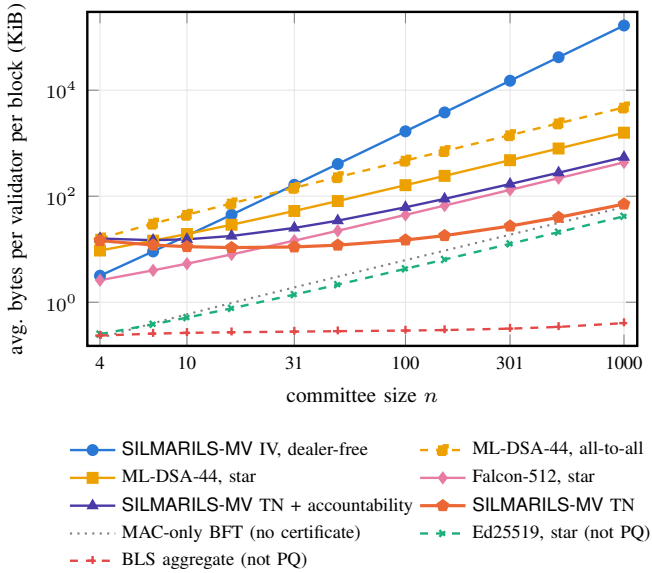

\begin{table}[!t]
\centering
\caption{Scalability ($n=3t+1$), analytic normal-case model. Per-validator average and maximum (leader or combiner) bytes per block, certificate size, and time for an external party to verify one certificate. TN: $\omega=16$, $\psi=64$, $\nu=64$.}
\label{tab:scale}
\scriptsize
\setlength{\tabcolsep}{2.2pt}
\begin{tabular}{@{}llrrrr@{}}
\toprule
 & & $n=31$ & $n=100$ & $n=301$ & $n=1000$\\
\midrule
\multirow{6}{*}{\shortstack[l]{avg.\ per\\ validator\\ (KiB)}}
 & ML-DSA-44 star & 52.7 & 161 & 478 & 1\,579\\
 & ML-DSA-44 all-to-all & 144 & 470 & 1\,420 & 4\,724\\
 & Falcon-512 star & 14.5 & 44.4 & 132 & 435\\
 & \scheme{} IV & 164 & 1\,675 & 15\,085 & 166\,160\\
 & \scheme{} TN & \textbf{11.0} & \textbf{14.9} & \textbf{27.4} & \textbf{71.0}\\
 & \scheme{} TN+acc. & 25.2 & 61.8 & 170 & 545\\
\midrule
\multirow{3}{*}{\shortstack[l]{max.\ party\\ (KiB)}}
 & ML-DSA-44 star & 1\,633 & 16\,146 & 143\,926 & 1\,579\,456\\
 & Falcon-512 star & 449 & 4\,443 & 39\,609 & 434\,677\\
 & \scheme{} TN & \textbf{59.9} & \textbf{101} & \textbf{221} & \textbf{636}\\
\midrule
\multirow{3}{*}{\shortstack[l]{certificate\\ (KiB)}}
 & ML-DSA-44 & 49.6 & 158 & 475 & 1\,576\\
 & Falcon-512 & 13.7 & 43.6 & 131 & 434\\
 & \scheme{} TN & \textbf{0.53} & \textbf{0.53} & \textbf{0.53} & \textbf{0.53}\\
\midrule
\multirow{3}{*}{\shortstack[l]{cert.\\ verify}}
 & ML-DSA-44 & 0.49\,ms & 1.55\,ms & 4.66\,ms & 15.5\,ms\\
 & Falcon-512 & 0.61\,ms & 1.96\,ms & 5.87\,ms & 19.5\,ms\\
 & \scheme{} TN & \textbf{2.4}\,$\mu$s & \textbf{2.4}\,$\mu$s & \textbf{2.5}\,$\mu$s & \textbf{2.4}\,$\mu$s\\
\bottomrule
\end{tabular}
\end{table}

\textbf{Bandwidth.} \Cref{fig:scale} and \cref{tab:scale} show the effect of the three mechanisms. IV mode grows as $n^2$ and is the right choice only for small consortia. TN grows only through the MAC tags of consensus, which any all-to-all protocol pays: at $n=1000$, TN costs $71$\,KiB per validator per block, of which $62.5$\,KiB are MAC tags and about $8$\,KiB are issuance. This is $22\times$ less than ML-DSA-44 in a star topology, $67\times$ less than ML-DSA-44 all-to-all, and $6\times$ less than Falcon-512. With information-theoretic MACs, consensus also consumes a pre-shared $32$-byte pad per tag and endpoint, $125$\,KiB per validator per block at $n=1000$. If this pad material is counted as traffic, TN totals $196$\,KiB, still $8\times$ below ML-DSA-44; with computational post-quantum MACs, it is not consumed at all (\cref{sec:discussion}). The busiest party is where the difference is largest. An ML-DSA-44 leader must send a $1.5$\,MiB QC to every replica, i.e., $1.5$\,GiB per block at $n=1000$, whereas the busiest TN party (the combiner) handles $636$\,KiB. For $n=4$, TN is more expensive than ML-DSA-44 in a star topology, because issuing keys to $64$ external verifiers dominates; IV mode is cheaper there. Adding accountability (TN+acc.) costs $545$\,KiB per validator per block at $n=1000$, still $2.9\times$ below ML-DSA-44 in a star topology, with attributable votes and $32$-byte vote signatures.

\textbf{Certificates.} A TN certificate is $544$\,B at every committee size, and an external verifier checks it in $2.4\,\mu$s including hashing. An ML-DSA-44 QC at $n=1000$ is $1.5$\,MiB and takes $15.5$\,ms to verify. Only BLS aggregates are smaller ($96$\,B plus a bitmap), and they are not post-quantum.

\textbf{Computation.} Measured in C over $\mathbb{F}_{2^{255}-19}$ ($\omega=16$, $\psi=64$), producing a share costs about $15\,\mu$s per member and slot. Combining is dominated by computing $t+1$ Lagrange coefficients: $12\,\mu$s at $n=31$, $0.20$\,ms at $n=301$ and $2.1$\,ms at $n=1000$, including the check with the combiner's own key. The fast decoder of \cref{fig:tn}, Gao decoding of one random combination followed by interpolation from non-suspects, takes $0.17$\,ms, $8.7$\,ms and $99$\,ms at the same sizes. It runs only when a share is wrong, and suspects are excluded from the optimistic path afterwards. We estimate that a TN validator spends about $1.5$\,ms per block at $n=1000$: $0.4$\,ms on MAC tags, assuming $0.2\,\mu$s per tag, and the rest on local issuance arithmetic. ML-DSA-44 verification alone takes $15.5$\,ms (star) or $46$\,ms (all-to-all).

\textbf{Clients.} An arbiter-verified client signature is $32$\,B. Its issuance costs each validator $0.38$\,KiB per transaction, independent of $n$, whereas a ML-DSA-44 signature delivers $2.4$\,KiB to every validator and costs each of them $23\,\mu$s to verify. The client receives about $32n$ bytes of key material per transaction: $3.2$\,KiB at $n=100$ and $31.7$\,KiB at $n=1000$.

\textbf{Simulation of the protocols.} The simulator (\texttt{tn\_sim.py} in~\cite{SILMVcode}) executes $\Pi_{\mathsf{TN}}$ end to end for $n\in\{4,7,10,13,16\}$, with $8$ committee indices, $\psi=4$ and $4$ external verifiers. The $t$ corrupted members are assigned, in turn, three behaviors: random shares, silence, and shares on a conflicting value (all three occur for $n\ge10$). All $32$ certificates per run were valid and accepted by every verifier. Combining used the optimistic path, with at most one coordinate-wise Berlekamp--Welch decoding per run, after which the combiner excluded the suspects it identified. The measured issuance cost matches \eqref{eq:tncost} to within $0.8\%$. A verifier enrolled after all slots had passed verified every historical certificate (\cref{prop:enrol}). Redistribution to a committee of $n+3$ members (\cref{lem:reshare}), with corrupted old dealers sharing wrong values and silent dealers, identified exactly the deviating dealers in every run. The new committee then issued and certified on a new index, and all old certificates still verified. For arbiter-verified keys, with the weights $r^{j-1}$ of \cref{lem:batch}, batches of $24$ valid signatures were always accepted, and batches with one modified signature were always rejected. Separately, an exhaustive enumeration over $\mathbb{F}_3$, $\mathbb{F}_5$ and $\mathbb{F}_7$ (\texttt{batch\_check.py}) reproduces the counterexample to the weights $r^j$ given after \cref{lem:batch} (acceptance $19/27$, $61/125$ and $127/343$ for $L=1$) and confirms the bound $(L+1)/q$ of the corrected lemma for batches of $L\in\{1,2\}$ unsigned honest items and every choice of the submitted values. Finally, in the TN forgery experiment ($n=7$, $\omega=2$, $\psi=2$), $t$ corrupted members and $\omega$ corrupted verifiers see $\psi-1$ released certificates on an index and try to certify a value that the committee did not decide. The adversary samples the signing matrix from its exact posterior given the corrupted keys and the certificates; by \cref{lem:sharesim}, the member shares add nothing. Over $2\times10^4$ trials it succeeded with probability $0.0079$ at $q=257$ (bound $0.0078$, standard error $0.0006$, i.e., within one standard error of the bound) and $0.0189$ at $q=101$ (bound $0.0199$). These runs support \cref{thm:tn}(i) but do not prove it.

\subsection{HotStuff Prototype}\label{sec:eval:proto}
To complement the analytic model, we implemented chained HotStuff~\cite{HotStuff2019} in Rust (tokio). It has a three-chain commit rule, round-robin leaders and a timeout-based view change, with pluggable vote authentication and real cryptography throughout:
Ed25519 (\texttt{ed25519-dalek}); ML-DSA-44 (PQClean, AVX2); \scheme{} in IV mode, with aggregated QCs; MACs over $\mathbb{F}_{2^{255}-19}$ with fixed per-pair keys (a timing model of the analyzed polynomial MAC); MACs plus TN certificates ($\omega=16$, $\psi=64$), sent to $4$ external auditors.
Star-topology variants send votes to the next leader. All-to-all variants (``a2a'') broadcast votes, and every replica forms the QC locally. In the MAC variants a QC is not transferable: a replica accepts a QC, whether it arrives in a proposal or in a new-view message, only after it has itself received $2t+1$ valid MAC-authenticated votes for that block and view, and it holds the proposal until then. Since every honest replica votes once per view and sends the same vote to everyone, two locally formed QCs for conflicting blocks in one view would need an honest replica that voted twice; hence the quorum-intersection step of the HotStuff safety argument~\cite{HotStuff2019} carries over to locally formed QCs, provided the MACs are unforgeable, i.e., use pads as in \cref{sec:tn}, which the timing prototype does not. Liveness under Byzantine voters that withhold votes from some replicas is not addressed by this rule; a deployment would use a MAC-based protocol with a proved view change, such as PBFT~\cite{PBFT2002} (\cref{sec:tn}). For timing, the prototype derives the MAC keys from a fixed per-pair key instead of consuming pads; the pad material is accounted analytically in \cref{sec:eval:scale}.
\scheme{} votes and QCs carry their $4$-byte indices, which the wire sizes include. All replicas run on one host and communicate over an emulated network. The network applies a per-link one-way delay and a per-node egress bandwidth cap, charges every message its exact wire size, and delivers through one delay queue per node. Blocks carry $400$ transactions of $128$\,B. Keys are issued before the measured window; issuance traffic is either omitted or injected as background traffic following \eqref{eq:dkicost} and \eqref{eq:tncost}. Unless stated otherwise, the one-way delay is $10$\,ms and the egress cap is $1$\,Gbit/s.

All experiments ran on the host of \cref{sec:eval} ($24$ cores, $32$ threads). Each configuration was run three times for $15$\,s after a $4$\,s warm-up, and we report medians. In all $354$ runs, every honest replica committed the same chain. These runs involve honest and crashed replicas only; they are a performance evaluation and do not test Byzantine safety under equivocation, selective delivery or fabricated QCs.

\begin{figure}[!t]
\centering
\begin{tikzpicture}
\begin{groupplot}[group style={group size=1 by 2, vertical sep=1.3cm}, width=\columnwidth, height=4.4cm,
  xmode=log, xtick={4,10,31,100}, xticklabels={4,10,31,100}, xmin=3.6, xmax=112, legend cell align=left]
\nextgroupplot[ylabel={throughput (ktx/s)}, ymin=0, ymax=19, y filter/.expression={y/1000}]
\addplot[cSIL, mark=*] table[x=x,y=siliv]{data_hs_E1_tput_tx_s.dat};
\addplot[cAGG, mark=pentagon*, very thick] table[x=x,y=mactn]{data_hs_E1_tput_tx_s.dat};
\addplot[cD44, mark=square*] table[x=x,y=mldsa44]{data_hs_E1_tput_tx_s.dat};
\addplot[cD44, mark=square*, dashed] table[x=x,y=mldsa44a2a]{data_hs_E1_tput_tx_s.dat};
\addplot[cED, mark=x, dashed] table[x=x,y=ed25519]{data_hs_E1_tput_tx_s.dat};
\nextgroupplot[xlabel={committee size $n$}, ylabel={crypto CPU / replica / block (ms)}, ymode=log, ymin=0.002, ymax=20,
  legend columns=2, legend style={at={(0.5,-0.42)},anchor=north}]
\addplot[cSIL, mark=*] table[x=x,y=siliv]{data_hs_E1_crypto_ms_per_replica_per_block.dat}; \addlegendentry{\scheme{} IV (star)}
\addplot[cAGG, mark=pentagon*, very thick] table[x=x,y=mactn]{data_hs_E1_crypto_ms_per_replica_per_block.dat}; \addlegendentry{\scheme{} TN (a2a)}
\addplot[cD44, mark=square*] table[x=x,y=mldsa44]{data_hs_E1_crypto_ms_per_replica_per_block.dat}; \addlegendentry{ML-DSA-44 (star)}
\addplot[cD44, mark=square*, dashed] table[x=x,y=mldsa44a2a]{data_hs_E1_crypto_ms_per_replica_per_block.dat}; \addlegendentry{ML-DSA-44 (a2a)}
\addplot[cED, mark=x, dashed] table[x=x,y=ed25519]{data_hs_E1_crypto_ms_per_replica_per_block.dat}; \addlegendentry{Ed25519 (not PQ)}
\end{groupplot}
\end{tikzpicture}
\caption{HotStuff prototype, $10$\,ms one-way delay and $1$\,Gbit/s per node: throughput and cryptographic CPU time per replica per block versus committee size (medians of $3$ runs, i9-13900HX).}
\label{fig:proto}
\end{figure}
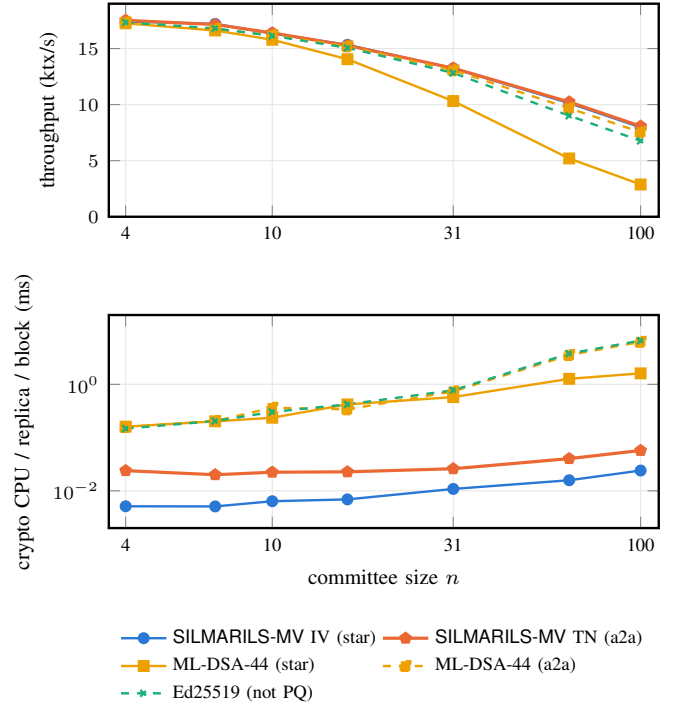

\begin{table}[!t]
\centering
\caption{HotStuff prototype (medians of $3$ runs, i9-13900HX). Throughput (ktx/s), commit latency (ms), KiB sent per replica per block (including the $50$\,KiB payload), and cryptographic CPU time per replica per block ($\mu$s, $n=31$).}
\label{tab:proto}
\scriptsize
\setlength{\tabcolsep}{3pt}
\resizebox{\columnwidth}{!}{%
\begin{tabular}{@{}l|rr|rr|rrr|r@{}}
\toprule
 & \multicolumn{2}{c|}{$n{=}31$, 1\,ms} & \multicolumn{2}{c|}{$n{=}31$, 100\,Mb/s} & \multicolumn{3}{c|}{$n{=}100$, 10\,ms, 1\,Gb/s} & \\
Scheme & tput & lat & tput & lat & tput & lat & KiB & crypto\\
\midrule
ML-DSA-44 & 18.4 & 68.7 & 2.08 & 602 & 2.88 & 429 & 209 & 576\\
ML-DSA-44 a2a & 30.8 & 38.7 & 3.57 & 331 & 7.55 & 158 & 288 & 737\\
Ed25519 (not PQ) & 28.6 & 45.3 & 3.60 & 349 & 6.77 & 193 & 54.0 & 772\\
\scheme{} IV & \textbf{32.1} & 39.7 & \textbf{3.68} & 342 & \textbf{7.97} & 162 & \textbf{52.0} & \textbf{11}\\
MAC a2a & 32.5 & 37.0 & 3.71 & 324 & 8.08 & 148 & 56.8 & 10\\
\scheme{} TN & \textbf{32.6} & 36.9 & \textbf{3.68} & 325 & \textbf{8.08} & 148 & 57.4 & 26\\
\bottomrule
\end{tabular}}
\end{table}

\Cref{fig:proto} and \cref{tab:proto} summarize the results.
\emph{Scaling.} At $n=100$, \scheme{} IV commits $7.97$\,ktx/s, against $2.88$\,ktx/s for ML-DSA-44 in the same star topology ($2.77\times$). Its latency is $162$\,ms instead of $429$\,ms, and each replica sends $52$\,KiB per block instead of $209$\,KiB. ML-DSA-44 recovers most of the throughput with all-to-all voting ($7.55$\,ktx/s), but it then sends $288$\,KiB per replica per block, $5.5\times$ more than \scheme{} IV. \emph{Cryptographic CPU.} Per replica and block, \scheme{} IV spends $24\,\mu$s at $n=100$ and TN spends $57\,\mu$s, including share computation and combining. ML-DSA-44 spends $1.61$\,ms in the star topology, $67\times$ more than \scheme{} IV. The values for Ed25519 and for ML-DSA-44 a2a at $n\ge64$ are inflated by contention among colocated replicas. \emph{Bandwidth-bound regime.} At $n=31$, \scheme{} IV sustains $1.85\times$, $1.77\times$ and $1.62\times$ the throughput of ML-DSA-44 at $50$, $100$ and $250$\,Mbit/s, because an ML-DSA-44 QC adds $50$\,KiB to every copy of the proposal. \emph{CPU-bound regime} ($1$\,ms delay, $n=31$). \scheme{} IV reaches $1.75\times$ the throughput of ML-DSA-44. \emph{Threshold notarization} costs less than $1\%$ throughput relative to MAC-only consensus at every committee size. External auditors hold verified certificates $22.4$--$23.0$\,ms after commit, independent of $n$ from $4$ to $100$, i.e., two message delays plus processing. \emph{Issuance traffic} (background, following \eqref{eq:dkicost} and \eqref{eq:tncost}). In TN mode it adds $6$\,KiB per replica per block at $n=100$ and leaves throughput unchanged. In IV mode it reaches $1.7$\,MiB per replica per block at $n=100$ and lowers throughput by $35\%$, which confirms that IV mode is meant for small committees. \emph{Crashes.} With $t$ crashed replicas, recovery is dominated by the timeout policy of our simple pacemaker, and all authentication schemes behave alike (within $10\%$). \emph{Sensitivity to colocation.} Varying the runtime from $8$ to $32$ threads at $n=100$ changes the throughput of \scheme{} IV, TN and star ML-DSA-44 by less than $0.4\%$, and that of all-to-all ML-DSA-44 by $7\%$.
The emulation shares one machine among all replicas and does not model packet loss or jitter. The per-replica cryptographic CPU time is the metric least affected by colocation.

\subsection{Empirical Validation of the Bounds}\label{sec:eval:security}
We implemented the adversaries of \cref{prop:tight} over small primes $q\in\{11,31,101,257\}$ with $(\omega,\psi)\in\{(1,1),(2,2),(3,4)\}$, using $4\times10^4$ trials per point:
the \emph{posterior forger} corrupts $\omega$ verifiers, obtains $\psi$ signatures on the target index, samples a signing matrix from the exact posterior given its view, and signs a fresh message; the \emph{repudiator}, a corrupted signer colluding with $\omega$ verifiers, submits $\balpha=A\bfe(m)+\bD$ with $\bD\ne0$ orthogonal to all corrupted functionals; the tightness attacks (b) and (c) use $\omega+1$ colluders and $\psi+1$ signatures, respectively.

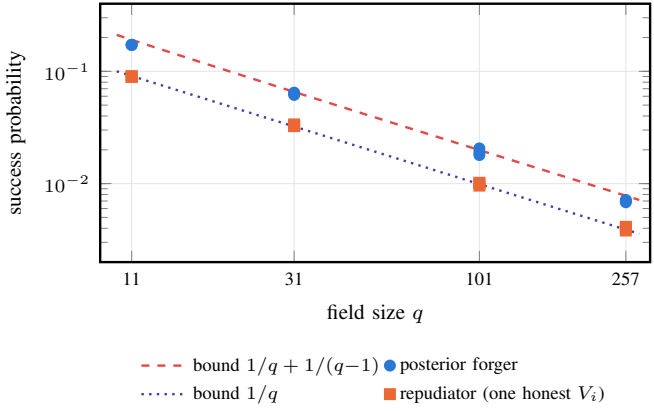
\begin{figure}[!t]
\centering
\begin{tikzpicture}
\begin{loglogaxis}[width=\columnwidth, height=5.0cm, xlabel={field size $q$}, ylabel={success probability},
  xmin=9, xmax=300, ymin=0.002, ymax=0.4, legend columns=2, legend style={at={(0.5,-0.32)},anchor=north}, legend cell align=left, xtick={11,31,101,257}, xticklabels={11,31,101,257}]
\addplot[cRED, dashed, no marks, domain=10:280, samples=40] {1/x + 1/(x-1)}; \addlegendentry{bound $1/q+1/(q{-}1)$}
\addplot[cSIL, only marks, mark=*] table[x=q,y=F1va]{data_attacks.dat}; \addlegendentry{posterior forger}
\addplot[cD65, dotted, no marks, domain=10:280, samples=40] {1/x}; \addlegendentry{bound $1/q$}
\addplot[cAGG, only marks, mark=square*] table[x=q,y=R1va]{data_attacks.dat}; \addlegendentry{repudiator (one honest $V_i$)}
\end{loglogaxis}
\end{tikzpicture}
\caption{Monte-Carlo success of the best attacks against the bounds of \cref{thm:main} ($Q=1$). Each $q$ has three points, for $(\omega,\psi)\in\{(1,1),(2,2),(3,4)\}$.}
\label{fig:attacks}
\end{figure}

\Cref{fig:attacks} plots the results against the bounds of \cref{thm:main} for $Q=1$. Forgery success is close to the bound $1/q+1/(q-1)$ and exceeds it by at most one standard error (the largest excess, $0.0206$ against $0.0199$ at $q=101$, is $0.95$ standard errors); at $q=257$, for example, we measure $0.0068$--$0.0072$ (standard error about $0.0004$) against a bound of $0.0078$. The second term is not vacuous: a singular $\Gamma$ \emph{can} place an honest functional in the corrupted row space (conditional on $\rank K=r$, this happens with probability $q^{\,r-\omega-1}$, not always), and the posterior forger then succeeds. Repudiation success equals $1/q$ within sampling error, matching \cref{prop:tight}(d). The independent-functional variant (no anchors) behaves the same, with forgery success at $2/q$, matching its $2Q/q$ bound. In that variant, $\omega+1$ colluders forged with probability $0.89$--$0.997$ ($2\,000$ trials per point), which is the probability that their independent random functionals are linearly independent. In the anchored scheme, the attack of \cref{prop:tight}(b) succeeded in \emph{every} one of $4\,000$ trials per field size, as predicted. With $\psi+1$ signatures, forgery succeeded in every trial (\cref{prop:tight}(c)). Finally, the $Q$-query attack of \cref{prop:tight}(a) without corruptions succeeded with probability $0.0045$, $0.016$ and $0.066$ at $q=257$ for $Q=1,4,16$, against bounds of $0.0078$, $0.019$ and $0.066$, consistent with \cref{prop:tight}(a), which shows that the leading term $Q/q$ is tight. Small-field experiments support the analysis; they do not prove the bounds or their tightness.

\section{Discussion and Limitations}\label{sec:discussion}

\textbf{Limitations.} \Cref{tab:limits} lists the limitations that \cref{sec:scale} removes and those that remain.

\begin{table}[!t]
\centering
\caption{Limitations removed by the scaling mechanisms, and limitations that remain.}
\label{tab:limits}
\scriptsize
\setlength{\tabcolsep}{2.5pt}
\renewcommand{\arraystretch}{1.1}
\begin{tabular}{@{}>{\raggedright\arraybackslash}p{0.25\columnwidth}>{\raggedright\arraybackslash}p{0.71\columnwidth}@{}}
\toprule
\textbf{Limitation} & \textbf{Status}\\
\midrule
\multicolumn{2}{c}{\cellcolor{gray!18}\textbf{Removed}}\\
Issuance cost & Bounded per validator independently of $n$ in TN mode \eqref{eq:tncost}; only consensus authentication grows with $n$.\\
\midrule[0.25pt]
Signature size & $32(\omega+1)$\,B, with $\omega$ chosen for external verifiers, not tied to $t$.\\
\midrule[0.25pt]
Signer state & TN indices follow consensus slots; an honest committee cannot reuse one.\\
\midrule[0.25pt]
Closed verifier set & Retroactive enrollment (\cref{prop:enrol}).\\
\midrule[0.25pt]
Membership changes & Robust redistribution with perfect privacy (\cref{lem:reshare}); fresh anchors in every epoch.\\
\midrule[0.25pt]
Clients & $32$-byte arbiter-verified signatures, constant cost per validator and transaction.\\
\midrule
\multicolumn{2}{c}{\cellcolor{gray!18}\textbf{Remaining}}\\
Verification not public & Secret keys and admission control, since $\omega+1$ colluders forge (\cref{prop:tight}(b)); rare external checkpoints can use ML-DSA or SLH-DSA.\\
\midrule[0.25pt]
Attribution in TN & Certificates attest the decision, not votes; accountable TN costs $545$\,KiB per validator per block at $n=1000$, below ML-DSA-44 but $7.7\times$ TN.\\
\midrule[0.25pt]
Key consumption & About one element per certified slot and signer; one-time MAC pads cost $125$\,KiB per validator per block at $n=1000$. Computational PQ MACs avoid this, but agreement then holds only against adversaries that cannot break the MAC \emph{during} the run; a failed agreement can exceed $\psi$ messages on one index and expose it (\cref{prop:tight}(c)).\\
\midrule[0.25pt]
Client key delivery & About $32n$ bytes per transaction.\\
\midrule[0.25pt]
Private channels & Needed for issuance (below).\\
\midrule[0.25pt]
Consensus integration & \Cref{thm:tn} assumes agreement of $\Pi_{\mathsf{BFT}}$ (e.g., PBFT with MAC authenticators); its digests add a collision term unless replicas agree on full values (\cref{rem:budget}). The MAC variant of our HotStuff prototype is a performance prototype, and its Byzantine liveness is not established. The traffic figures are an analytic normal-case model, not the cost of a complete deployment.\\
\midrule[0.25pt]
Long-term trust & Each refresh is covered by \cref{lem:reshare}, but security against a mobile adversary is not proved; $\omega$ bounds verifiers corrupted cumulatively per anchor matrix; historical evidence via later certificates relies on hashing and a trusted checkpoint.\\
\midrule[0.25pt]
Maturity & Malicious dealing checks of~\cite{BTH2008} are costed (as an estimate), not implemented; $n=1000$ results are analytic; the prototype runs on one host over an emulated network without loss or jitter, and only crash faults were tested.\\
\bottomrule
\end{tabular}
\end{table}

\textbf{Private channels for issuance.} Information-theoretic security of $\Pki$ and $\Pi_{\mathsf{TN}}$ needs private channels. They can be realized in three ways:
(i)~one-time pads from QKD links, as already deployed in quantum-secured ledger testbeds~\cite{Kiktenko2018QSB}, consumed offline at the rate given by \eqref{eq:dkicost} or \eqref{eq:tncost}; (ii)~pre-distributed pad material, e.g., by courier or HSM bootstrapping; (iii)~post-quantum TLS with forward secrecy as a pragmatic compromise.
With option (iii), signatures produced in an epoch are information-theoretically unforgeable against anyone who has not broken that epoch's issuance channels \emph{by then}. An attacker who records issuance traffic and later breaks the key exchange could learn old keys and forge signatures \emph{for old epochs}. Evidence that must remain valid indefinitely should therefore be anchored in later certificates, which in turn relies on collision resistance and on an authenticated checkpoint (\cref{sec:reshare}). Options (i) and (ii) avoid the issue entirely, which is the sense of ``everlasting'' security in~\cite{Unruh2013}. Signing and verification never need confidential channels.

\textbf{Threshold trust.} The committee acts as a \emph{distributed} authority: $t+1$ members who pool their shares of an epoch can forge for that epoch. Within an epoch this is the BFT assumption $t<n/3$. Proactive refresh and erasure (\cref{sec:reshare}) are designed so that it suffices that no $t+1$ members are corrupted \emph{within the same epoch}, so that members who leave cannot later pool erased shares; we prove each refresh step (\cref{lem:reshare}) but not security against a mobile adversary, and erasure cannot remove copies that corrupted members kept. This is the trust model of threshold-signature-based ledgers~\cite{CKS2005,SBFT2019}, and it is strictly weaker than trusting one authority. A malicious threshold of validators may refuse to issue new verifier credentials or selectively admit verifiers, but such behavior constitutes a denial of service at the governance layer rather than a forgery attack against \scheme{} authentications.

\textbf{Liveness.} $\Pki$ and $\Pi_{\mathsf{TN}}$ have guaranteed output for $t<n/3$ in synchronous issuance rounds. Certificate release inherits the liveness of the consensus protocol, and combining is robust also under asynchrony (\cref{thm:tn}(iii)). Asynchronous issuance and variants for $t<n/2$ with statistical security~\cite{RabinBenOr1989} are outside our scope.

\textbf{Relation to threshold signatures.} \Cref{sec:tn} and \cref{tab:threshold} compare TN with BLS and lattice threshold signatures.

\textbf{Relation to \silm.} \scheme{} complements~\cite{SILMARILS}. The two-party mode of \silm{} remains the right tool when \emph{deniability} is wanted, such as private user--validator sessions. \scheme{} is the right tool when \emph{evidence} is wanted. Its secret reconstruction weights $\bc_i$ are precisely what turns \silm's publicly reconstructed, simulatable algebraic core into a transferable, non-repudiable signature.

\section{Conclusion}\label{sec:conclusion}
We introduced \scheme{}, a multi-verifier unconditionally secure signature in which verifiers hold secret linear functionals derived from virtual anchors that nobody knows. The scheme is strongly unforgeable, non-repudiable and transferable without bound, with a bound of $Q/q+1/(q-1)$ against unbounded adversaries whose leading term is tight. It admits linear aggregation, and its keys are issued by the validator committee itself with perfect security. Letting the committee itself sign, through threshold notarization, makes certificates constant-size and per-validator issuance independent of the committee size. Together with arbiter-verified keys for clients, retroactive enrollment and robust redistribution, this brings information-theoretic authentication to committees of $1000$ validators at an estimated $71$\,KiB per validator per block in the normal case, $22\times$ below ML-DSA-44 within the admitted-verifier setting considered by \scheme{}.

Open directions include:
a full treatment of the linear-communication variant of \cref{rem:linear} and of asynchronous issuance; a malicious-secure implementation of issuance integrated with a production BFT engine; packed secret sharing to reduce client key delivery; lower bounds on key and signature length for multi-verifier USS with aggregation; integration with QKD networks for fully everlasting deployments; a proof of security against mobile adversaries across refreshes; and a MAC-based consensus engine integrated with TN, with Byzantine safety and liveness proved and tested in the deployed configuration.

\bibliographystyle{IEEEtran}
\bibliography{refs}

\appendices
\section{Reproducibility}\label{app:artifact}
The artifact is publicly available at \url{https://github.com/hassankhodaiemehr/silmarils-mv}~\cite{SILMVcode}. The files below are in its \texttt{silmarils-mv/} directory (subdirectories \texttt{code/}, \texttt{hsbench/} and \texttt{results/}). It contains:
\texttt{fe25519.h} and \texttt{silmv\_bench.c}: the C core and its benchmark; \texttt{silmv.py}: the reference implementation; \texttt{dki\_va.py} and \texttt{dki\_sim.py}: the issuance simulators, with and without anchors; \texttt{attacks.py}, \texttt{attacks\_va.py} and \texttt{attacks\_va\_tight.py}: the Monte-Carlo validation; \texttt{batch\_check.py}: the exhaustive check of batch arbitration (\cref{lem:batch}); \texttt{hsbench/run\_baselines.sh} and \texttt{make\_baselines.py}: the same-host AVX2 baselines, converted into \texttt{results/baselines.json}; \texttt{ledger\_model.py}: the ledger model, which generates every table and figure from the raw measurements. \texttt{tn\_bench.c}: threshold-notarization share, combine, fast-decoder and certificate-verification benchmark; \texttt{tn\_sim.py}: end-to-end simulation of $\Pi_{\mathsf{TN}}$ with Byzantine members, retroactive enrollment, robust redistribution, batch arbitration of arbiter-verified keys, and the TN forgery experiment (\texttt{TNQ=257}); \texttt{scale\_model.py}: the scalability model up to $n=1000$ (\cref{fig:scale,tab:scale}). \texttt{hsbench/}: the Rust HotStuff prototype of \cref{sec:eval:proto}, with \texttt{run\_matrix.py} (experiment matrix) and \texttt{summarize.py} (medians and tables).

\end{document}